\documentclass[a4paper]{article}
\usepackage{graphicx}
\usepackage{xspace}
\usepackage{color}
\usepackage{enumitem}
\usepackage[english]{babel}
\usepackage{amssymb}
\usepackage{amsthm}
\usepackage{amsfonts}
\usepackage{amsmath}
\usepackage{fullpage}
\usepackage{graphicx}
\usepackage{subfigure}
\usepackage{algorithm}
\usepackage[noend]{algorithmic}
\usepackage{amsmath,amssymb,cite}
\usepackage{lineno,footmisc,marvosym}
\usepackage{wasysym,vmargin,stackrel}
\usepackage{color,hyperref}
\usepackage{boxedminipage}
\usepackage{xspace}
\usepackage{fancyhdr}
\usepackage{tcolorbox}
\usepackage{extarrows}
\usepackage{blindtext}

\providecommand{\ceil}[1]{\left \lceil #1 \right \rceil }
\providecommand{\floor}[1]{\left \lfloor #1 \right \rfloor }

\newtheorem{theorem}{Theorem}

\newtheorem{conjecture}{Conjecture}
\newtheorem{corollary}{Corollary}

\newtheorem{definition}{Definition}
\newtheorem{observation}{Observation}

\newtheorem{lemma}{Lemma}

\newcommand{\algorithmicbreak}{\textbf{break}}
\newcommand{\BREAK}{\STATE \algorithmicbreak}

\newcommand{\csd}{\textsc{CSD}\xspace}

\newcommand{\trpartition}{\textsc{3-Partition}\xspace}
\newcommand{\NP}{\ensuremath{\mathtt{NP}}\xspace}
\newcommand{\coverage}{\ensuremath{\text{coverage}}}
\newcommand{\supp}[1]{\text{supp}\ensuremath{(#1)}}
\newcommand{\DR}[2]{\text{DR}\ensuremath{(#1,#2)}}
\newcommand{\PoD}{\ensuremath{\mathrm{PoD}}}

\newcount\Comments  
\definecolor{darkgreen}{rgb}{0,0.6,0}
\newcommand{\kibitz}[2]{\ifnum\Comments=1{\color{#1}{#2}}\fi}

\title{\vspace{-0.5cm}Connected Subgraph Defense Games\thanks{Supported in part by the NeST initiative of the School of EEECS of the University of Liverpool.}\vspace{0.5cm}}

\author{Eleni C.~Akrida\thanks{Department of Computer Science, University of Liverpool, Liverpool, UK.}
	\and Argyrios Deligkas\thanks{Department of Computer Science, University of Liverpool, UK and 
		Leverhulme Research Centre for Functional Materials Design, Liverpool, UK.}
	\and Themistoklis Melissourgos\thanks{Department of Computer Science, University of Liverpool, Liverpool, UK.} 
	\and Paul G.~Spirakis\thanks{Department of Computer Science, University of Liverpool, UK and 
		Computer Engineering \& Informatics Department, University of Patras, Greece. } \\
	\texttt{\small \{E.Akrida,Argyrios.Deligkas,T.Melissourgos,P.Spirakis\}@liverpool.ac.uk}}
\date{}

\begin{document}
	
\maketitle              

\begin{abstract}
We study a security game over a network played between a {\em defender} and $k$ {\em
attackers}. 
Every attacker chooses, probabilistically, a node of the network to damage. 
The defender chooses, probabilistically as well, a connected induced subgraph of the network 
of $\lambda$ nodes to scan and clean.
Each attacker wishes to maximize the probability of escaping her cleaning by the defender.
On the other hand, the goal of the defender is to maximize the expected number of attackers 
that she catches. 
This game is a generalization of the model from the seminal paper of Mavronicolas et al.
\cite{MavronicolasMPPS06}. 
We are interested in Nash equilibria of this game, as well as in characterizing \emph{defense-optimal} networks which allow for the best \emph{equilibrium defense ratio}, termed {\em Price of Defense}; this is the ratio of $k$ over the expected number of attackers that the defender catches in equilibrium. 
We provide characterizations of the Nash equilibria of this game and defense-optimal networks.
This allows us to show that the equilibria of the game coincide independently from  the 
coordination or not of the attackers.
In addition, we give an  algorithm for computing Nash equilibria.  Our algorithm requires 
exponential time in the worst case, but it is polynomial-time for $\lambda$ constantly close to 1 or $n$. 
For the special case of tree-networks, we further refine our characterization which allows us to 
derive a polynomial-time algorithm for deciding whether a tree is defense-optimal and if this is
the case it computes a defense-optimal Nash equilibrium. On the other hand, we prove that it
is \NP-hard to find a best-defense strategy if the tree is not defense-optimal. 
We complement this negative result with a polynomial-time constant-approximation
algorithm that computes solutions that are close to optimal ones for general graphs. 
Finally, we provide asymptotically (almost) tight bounds for the Price of Defense for any $\lambda$.



~

\end{abstract}

\section{Introduction}\label{sec:intro}

With technology becoming a ubiquitous and integral part of our lives, we find ourselves using several different types of ``computer'' networks. 
An important issue when dealing with such networks, which are often prone to security breaches~\cite{Cheswick03}, is to prevent and monitor unauthorized access and misuse of the network or its accessible resources. Therefore, the study of network security has attracted a lot of attention over the years~\cite{Stallings03}. Unfortunately, such breaches are often inevitable, since some parts of a large system are expected to have weaknesses that expose them to security attacks; history has indeed shown several
successful and highly-publicized such incidents~\cite{Spafford1989}.
Therefore, the challenge for someone trying to keep those systems and networks of computers secure is to counteract these attacks
as efficiently as possible, once they occur.

To that end, inventing and studying appropriate theoretical models that capture the essence of the problem is an important line of research,
ongoing for a few years now~\cite{MavronicolasPPS05, MavronicolasPPPS06}. In this work, extending some known models for very simple cases attacks and defenses~\cite{MavronicolasMPPS06, MavronicolasPPS08}, we introduce 
and analyze a more general model for a scenario of network attacks and defenses modeling it as a \emph{defense game}.

\paragraph{The Network Security Game.}

We follow the terminology established by the seminal paper of Mavronicolas et al.~\cite{MavronicolasPPS08}.
We consider a network whose nodes are vulnerable to infection by threats called \emph{attackers};
think of those as viruses, worms, Trojan horses or eavesdroppers~\cite{FranklinGY00} infecting the components of a computer network.
Available to the network is a security software (or firewall), called the \emph{defender}. 
The defender is only able to ``clean'' a limited
part of the network from threats that occur; the reason for the limited cleaning capacity of the defender
may be, for example, the cost of purchasing a global security software.
The defender seeks to protect the network as much as possible, and on the other hand,
every attacker seeks to increase the likelihood of not being caught.
Both the attackers and the defender make individual decisions for their positioning
in the network with the aim to maximize their own objectives. 

Every attacker targets (and attacks) a node chosen via her own probability distribution over the nodes of the network.
The defender cleans a connected induced subgraph of the network with size $\lambda$,
chosen via her own probability distribution over all connected induced subgraphs of the graph
with $\lambda$ nodes.
The attack of a particular attacker is successful unless the node chosen by the attacker is incident to an edge (link) being
cleaned by the defender, i.e.~to an edge belonging in the induced subgraph chosen by the defender.
One could equivalently think of the defender selecting a set of $\lambda$ connected nodes to defend, and an attacker is successful if and
only if she attacks a node that is not being defended.
Since attacks and defenses over a large computer network are self-interested procedures that
seek to maximize damage and protection, respectively, it is natural to model this network
security scenario as a non-cooperative \emph{strategic game} on graphs with two kinds
of players: $k \geq 1$ \emph{attackers}, each playing a \emph{vertex} of the graph, and a single \emph{defender} playing a \emph{connected induced subgraph} of the graph. The \emph{(expected) payoff} of an attacker is the probability that
she is not caught by the defender; the \emph{(expected) payoff} of the defender is
the (expected) number of attackers she catches.
We are interested in the Nash equilibria~\cite{Nash1951,Nash48} associated with this graph theoretic
game, where no player can unilaterally improve her (expected) payoff by switching to another probability distribution. We are also interested in understanding and characterizing the networks
that allow for a good \emph{defense ratio}: given a strategy profile, i.e.~a combination of strategies for the network entities (attackers and defender),
the defense ratio of a network is the ratio of the total number of attackers over the defender's expected payoff in that strategy profile.

\subsection{Our results}\label{sec:results}

In this paper we depart from and significantly extend the line of work of Mavronicolas et al. 
in their seminal paper~\cite{MavronicolasPPS08} on defense games in graphs; we term the type of games we consider
{\em CSD games}.
In our model the defender is more powerful than in~\cite{MavronicolasPPS08}, since her power
is parameterized by the size, $\lambda$, of the defended part of the network. We allow 
$\lambda$ to take values from 1 to $n$, while in~\cite{MavronicolasPPS08} only the case 
where $\lambda=2$ was studied.
We study many questions related to CSD games.
%
We extend the notions of \emph{defense ratio} and \emph{defense-optimal graphs} for 
CSD games.
In fact, the defense ratio of a given graph $G$ and a given strategy profile $S$ of the attackers and the defender is the ratio of the number of attackers, $k$, 
over the defender's expected payoff (the number of attackers she catches on expectation).
We thoroughly investigate the notion of the defense ratio for Nash equilibria strategy profiles.

Firstly, we precisely characterize the Nash equilibria  and defense-optimal graphs in CSD
games. 
This allows us to show that, in equilibrium, the game version of $k$ uncoordinated attackers
and a single defender is equivalent to the version in which a single leader coordinates the $k$
attackers, meaning that both versions of the game have the same defense ratio. We present an LP-based algorithm to compute an exact equilibrium of any given CSD game, whose running time is polynomial in $\binom{n}{\lambda}$.
Then, we focus on tree-graphs. There, we further refine our equilirbium characterization
which allows us to derive a polynomial-time algorithm for deciding whether a tree is 
defense-optimal and, if this is the case, it computes a defense-optimal Nash equilibrium. A tree
is defense-optimal if and only if it can be partitioned into $\frac{n}{\lambda}$  disjoint 
sub-trees.
On the other hand, we prove that it is \NP-hard to find a best-defense strategy if
the tree is not defense-optimal. 
%
We remark that a very crucial parameter for defense-optimality of a graph $G$ is the ``best'' probability with which any vertex of $G$ is defended in a NE; we call that probability \emph{MaxMin probability} and denote it by $p^*(G)$.
Then, for any graph $G$, the defense ratio in equilibrium is shown to be exactly $\frac{1}{p^*(G)}$.
Although it is hard to exactly compute $p^*(G)$ even for trees, we complement this negative result with a polynomial-time constant-approximation
algorithm that computes solutions that are close to the optimal ones for any $\lambda$, for any given general graph. 
In particular, we approximate the (best) defense ratio of any graph within a factor of 
$2+\frac{\lambda - 3}{n}$. Finally, we provide asymptotically tight bounds for the Price of Defense for any $\lambda \in \omega(1) \cap o(n)$, and almost tight bounds for any other value of $\lambda$.

\subsection{Related work}\label{sec:rel_work}

Our graph-theoretic game is a direct generalization of the defense game considered by Mavronicolas et al.~\cite{MavronicolasMPPS06,MavronicolasPPS08}. 
In the latter, the authors examined the case where the size of the defended part of the network is $\lambda = 2$, i.e.~where the defender ``cleans'' an edge.
This lead to a nice connection between equilibria and (fractional) matchings in the graph~\cite{MavronicolasPPPS06}.
But when $\lambda$ is greater than 2, one has to investigate (as we shall see here) how to sparsely cover the graph by as small
a number as possible of connected induced subgraphs of size $\lambda$.
This direction can be seen as an extension of fractional matchings to covers of the graph by equisized connected subgraphs.	
Sparse covering of graphs by connected induced subgraphs (clusters), not necessarily equisized, is a notion known to be useful also for distributed algorithms,
since it affects message communication complexity~\cite{Attiya2004}.


In another line of work, Kearns and Ortiz~\cite{KearnsO03} study \emph{Interdependent Security games} in which
a large number of players must make individual decisions regarding security. 
Each player's \emph{safety} may depend on the actions of the
entire population (in a complex way). The graph-theoretic game that we consider could be seen as a particular instance of
such games with some sort of limited interdependence: the actions of the defender and an attacker are interdependent, while the
actions of the attackers are not dependent on each other.

Aspnes et al.~\cite{AspnesCY06} consider a graph-theoretic game that models containment of the spread of viruses on a network;
each node individually must choose to either install anti-virus software at some cost, 
or risk infection if a virus reaches it without being stopped by some intermediate node
with installed anti-virus software. 
Aspnes et al.~\cite{AspnesCY06} prove several algorithmic properties for their
graph-theoretic game and establish connections to a certain graph-theoretic problem
called \emph{Sum-of-Squares Partition}.

A game on a weighted graph with two players, the \emph{tree player} and the \emph{edge player}, was studied by Alon et al.~\cite{AlonKPW95}.
At each play, the tree player chooses a spanning tree and the edge player chooses an edge of the graph, and the payoffs of the players
depend on whether the chosen edge belongs in the spanning tree. Alon et al. investigate the theoretical aspects of the above game and
its connections to the \emph{$k$-server problem} and \emph{network design}.
%

Finally, there is a long line of work on security games~\cite{an2011guards} where many 
scenarios are modelled using graph theoretic problems~\cite{jain2013,letchford2013,vanvek2012,Xu2016}.

\section{Preliminaries}\label{sec:prel}

\paragraph{The game.}
A {\em Connected-Subgraph Defense (CSD) game} is defined by a graph $G=(V,E)$, a 
{\em defender}, $k \geq 1$ {\em attackers}, and a positive integer $\lambda$. 
Throughout the paper, $\lambda$ is considered to be a \emph{given} parameter of the game.
A pure strategy 
for the defender is any induced connected subgraph $H$ of $G$ with $\lambda$ vertices, 
which we term {\em $\lambda$-subgraph}. 
For any $\lambda$-subgraph $H$ of $G$ we denote $V(H)$ its set of vertices.
Since $V(H)$ uniquely defines an induced subgraph of $G$, we will use the term 
$\lambda$-subgraph to denote either $V(H)$ or $H$. 
The {\em action set} of the defender is $D := \{ V(H) |  H \text{ is a } \lambda\text{-subgraph of } G \}$ and we will denote its cardinality by $\theta$, i.e. $\theta := |D|$. 
For ease of presentation, we will also refer to $D$ as $[\theta] := \{ 1, 2, \dots, \theta \}$.
A pure strategy for each of the attackers is any vertex of $G$.
So, the action set of every attacker is $V$, the vertex set of $G$; we denote $n := |V|$ and we similarly refer to $V$ also as $[n]$.

To play the game, the defender chooses a {\em defense (mixed) strategy}, i.e. a probability
distribution over her action set, and each attacker 
chooses an {\em attack (mixed) strategy}, i.e. a probability distribution over the vertices of 
$G$. 
We denote a strategy by $s := (s_1, \dots, s_{d}) \in \Delta_{d}$, i.e. by the probability distribution over $d$ enumerated pure strategies, where $\Delta_{d} := \{ x_1, \dots, x_d \geq 0 | \sum_{i=1}^{d} x_{i} = 1 \}$ is the $(d-1)$-unit simplex. In a defense strategy $q \in \Delta_{\theta}$ each pure strategy $j \in [\theta]$ is assigned a probability $q_j$. 

We say that a pure strategy of the defender, i.e. a specific $\lambda$-subgraph $H$ of $G$, \emph{covers} a vertex $v\in V$ if $v \in V(H)$.
A defense strategy covers a vertex $v\in V$ if it assigns strictly positive probability to at 
least one $\lambda$-subgraph $H$ of $G$ which contains $v$.

\begin{definition}[Vertex Probability]
	The {\em vertex probability} $p_{i}$ of vertex $i \in [n]$, is the probability that $i$ will 
	be covered, formally~	$p_{i} := \sum_{j \in [\theta]:~i \in j} q_j$.
\end{definition}
The {\em support} of a strategy $s$, denoted by \supp{s}, is the subset of the action set that is assigned strictly positive probability.

\vspace*{-1em}
\paragraph{Payoffs and Strategy profiles.}
A {\em strategy profile} is a tuple of strategies $S=(q, t_1, \dots, t_k)$, where $q$ denotes 
the defender's strategy and $t_j$ denotes the $j$-th attacker's strategy,~$j \in [k]$. A 
strategy profile is pure if the support of every strategy has size one.
The {\em payoff} of every attacker is 1 in any pure strategy profile where she does not choose a defended vertex, and 0 in all the rest. 
The payoff of the defender in a pure strategy profile where she defends $V(H)$, is the number of attackers that choose a vertex in $V(H)$. 
%
Under a strategy profile, the {\em expected payoff} of the defender is the expected number of attackers that she catches, which we call {\em defense value}, and the expected payoff of the attacker is the probability that she will not get caught.
A {\em best response} strategy for a participant is a strategy that maximizes her expected
payoff, given that the strategies of the rest of the participants are fixed.
A {\em Nash equilibrium}  is a strategy profile where all the participants are 
playing a best response strategy. In other words,  neither the defender nor any of the 
attackers can increase their expected payoff by unilaterally changing their strategy.



\begin{definition}[Defense Ratio]
	For a given graph $G$ we define a measure of the quality of a strategy profile $S$, called {\em defense ratio of $G$} and denoted \DR{G}{S}, as the ratio of the total number of attackers $k$ over the defense value. 
\end{definition}

In this work we are only interested in the cases where $S$ is an equilibrium.
 For a given graph, when in equilibrium, the defender's expected payoff is unique (due to Theorem \ref{lem:char_DR} and Corollary \ref{cor:from_NE_char} \ref{cor:unique_NE}) and achieves the {\em equilibrium defense ratio} \DR{G}{S^*}, where $S^*$ is an equilibrium.  The defense strategy in $S^*$ which achieves this  defense ratio will be termed {\em best-defense strategy}. 

\begin{definition}[MaxMin Probability, $p^*$]
	We call \emph{MaxMin Probability} of a graph $G$ the maximum, over all defense strategies, minimum vertex probability in $G$, that is:
	\[p^*(G) := \max_{q \in \Delta_{\theta}} \min_{i \in [n]} p_i.\]
\end{definition}

As we will show in Lemma \ref{lem:char_DR}, the equilibrium defense ratio of a graph $G$ turns out to be $\DR{G}{S^*} = 1/p^*(G)$.

\begin{definition}[Price of Defense]
	 The {\em Price of Defense}, \PoD, for a given parameter $\lambda$ of the game, is the worst defense ratio, over all graphs, achievable in equilibrium, that is:
	 \[\mathrm{PoD} (\lambda) = \max_{G} \mathrm{DR}(G,S^*) . \]
\end{definition}

 
\begin{definition}[Defense-Optimal Graph]
  For a given $\lambda$, a graph $G^*$ that achieves the minimum equilibrium defense ratio over all graphs, i.e.~$G^* \in \arg \min_{G} \mathrm{DR}(G,S^*)$, is called {\em defense-optimal graph}.
\end{definition} 

In the following, for ease of presentation, whenever we refer to defense optimality, we implicitly assume that $\lambda$ has a fixed value.

\section{Nash equilibria}

In this section, we provide a characterization of Nash equilibria in CSD games, as well as important properties of their structure which prove useful for the development of our algorithms in the remainder of the paper.

\begin{theorem}[Equilibrium characterization]\label{thm:char_NE}
	For a given graph $G$, in any equilibrium with support $S \subseteq [\theta]$ of the defender and support $T_j \subseteq [n]$ of each attacker $j \in [k]$, the following conditions are necessary and sufficient:
	\begin{enumerate}
		\item $\min_{i \in [n]} p_i$ is maximized over all defense strategies, and
		\item $\bigcup_{j \in [k]} T_j \subseteq V^*$, where $V^* := \arg \max_{q \in \Delta_{\theta}}min_{i \in [n]} p_i$,  and
		\item every $s \in S$ has the maximum expected total number of attackers on its vertices over all pure strategies.
	\end{enumerate}
\end{theorem}

\begin{proof}
First we will prove that the conditions in the statement of the lemma hold in equilibrium, i.e.~ equilibrium is sufficient for the conditions to hold. 

\textbf{Condition 1.} By definition, in an equilibrium the defender and each attacker have chosen a best response. Suppose that the defender has chosen some strategy $q = (q_1, q_2, \dots, q_\theta)$ over her action set $[\theta]$, and we will consider this strategy to be a vector variable for now. Given $q$, each vertex $i \in [n]$ has a vertex probability $p_i$. Now consider the minimum vertex probability $p' := \min_{i \in [n]} p_i$, and the set $V' \subseteq V$ consisting of the vertices with vertex probability $p'$, i.e. $V' := \arg \min_{i \in [n]} p_i$. Since an attacker plays a best response, her support will be a subset of $V'$; otherwise, if she assigns probability $t_v > 0$ on a vertex $v \notin V'$ (with $p_v > p'$) her expected payoff (see quantity \eqref{eq:attacker_exp_payoff}) can be strictly increased by choosing to move all of $t_v$ to another vertex $u \in V$, thus increasing her expected payoff by $t_u(p_v - p')$. Therefore, every attacker's support will be a subset of $V'$.

Now suppose that there are $k \geq 1$ attackers and let us denote the set of attackers by $[k]$. We will denote by $t_{ji}$ the probability that the strategy of attacker $j \in [k]$ has assigned on vertex $i \in [n]$. The expected payoff of the defender is:

\begin{align}\label{eq:defender_exp_payoff}
\sum_{\substack{i \in [n]}} \left( p_{i} \sum_{j \in [k]} t_{ji} \right).
\end{align} 
	
Since as we argued above, in an equilibrium, each attacker's strategy has support that is subset of $V'$, the expected payoff of the defender will be 
	
\begin{align*}
\sum_{\substack{i \in V'}} \left( p_i \sum_{j \in [k]} t_{ji} \right) + \sum_{\substack{i \in V \setminus V'}} \left( p_{i} \sum_{j \in [k]} t_{ji} \right) = p' \cdot \sum_{\substack{i \in V'}} \left( \sum_{j \in [k]} t_{ji} \right) = p' \cdot \sum_{j \in [k]} \left( \sum_{\substack{i \in V'}} t_{ji} \right) = p' \cdot k,
\end{align*} 
where the first equality is due to the fact that $p_i = p'$ $\forall i \in V'$ and $t_{ji} =0$ $\forall i \in V\setminus V'$, and the last equality is due to the fact that the support of any strategy $t_j=(t_{j1}, \dots, t_{ji})$ of an attacker $j \in [k]$ is a subset of $V'$. 
In an equilibrium, the defender also plays a best response, i.e. she maximizes her expected utility. Therefore, given the above quantity, the defender in an equilibrium has expected utility $\max_{q \in \Delta_{\theta}} p' \cdot k$, and Condition 1 of the lemma's statement is satisfied. 

\textbf{Condition 2.} The proof is by contradiction. Assume an equilibrium profile where the defender has strategy $q = (q_1, \dots, q_{\theta})$ and there is an attacker, $a$, with strategy $t = (t_1, \dots, t_n)$ whose support includes vertex $v \in [n]$ with $p_{v} > p'$, where $p' := \min_{i \in [n]} p_i$. Then $a$'s expected payoff is 

\begin{align}\label{eq:attacker_exp_payoff}
\sum_{\substack{i \in V\\i \neq v}} t_{i} (1 - p_{i}) + t_{v} (1 - p_{v}).
\end{align} 
However, $a$ can increase her expected payoff by moving all her probability $t_{v}$ to a vertex $v'$ for which $P_{v'} = P'$, which contradicts the equilibrium assumption.

\textbf{Condition 3.} The proof is by contradiction. Suppose that in an equilibrium the defender has strategy $q^* \in \Delta_{\theta}$, where $\supp{q^*} := S$.
According to Condition 1, this strategy achieves $p^*(G)$, and let us define the set $V^* := \arg \max_{q \in \Delta_{\theta}} \min_{i \in [n]} p_i$. We denote by $N_i$ the random variable that indicates the number of attackers on vertex $i \in [n]$, under the strategy profile determined by the strategy of the defender and each attacker. The expected utility of the defender is as in \eqref{eq:defender_exp_payoff}, or equivalently, $\sum_{\substack{i \in [n]}} \left( p_{i} \cdot \mathbb{E}[N_i] \right)$. Since, as argued above, in an equilibrium each attacker has support in $V^*$, the defender's expected payoff is in fact $p^* \cdot \sum_{\substack{i \in V^*}} \mathbb{E}[N_i]$. 

For the sake of contradiction, suppose that for the expected total number of attackers on two different pure defense strategies $s_1 \in S$ and $s_2 \in [\theta]$ it holds that $ \mathbb{E} \left[ \sum_{i \in s_1\phantom{j}} N_i \right]  < \mathbb{E}\left[\sum_{j \in s_2} N_j\right]$, and equivalently $\mathbb{E}\left[\sum_{i \in s_1 \setminus s_2} N_i\right] < \mathbb{E}\left[\sum_{j \in s_2 \setminus s_1} N_j\right]$. Then, the expected payoff of the defender can be strictly increased if she chooses a strategy $q' = (q'_1, \dots, q'_{\theta})$ where $q'_{s_1} = 0$ and $q'_{s_2} = q_{s_2}^* + q_{s_1}^*$. In particular, when the defender plays $q^*$ her expected payoff is
	 
\begin{align*}
U^* = p^* \cdot \mathbb{E}\left[ \sum_{\substack{i \in V \setminus (s_1 \cup s_2)}} N_i \right] + p^* \cdot \mathbb{E}\left[ \sum_{\substack{j \in s_1 \cap s_2}} N_j\right] + p^* \cdot \mathbb{E}\left[ \sum_{\substack{l \in s_2 \setminus s_1}} N_l\right] + p^* \cdot \mathbb{E}\left[ \sum_{\substack{r \in s_1 \setminus s_2}} N_r\right],
\end{align*}
whereas when she plays $q'$ it is
	 
\begin{align*}
U' &= p^* \cdot \mathbb{E}\left[ \sum_{\substack{i \in V \setminus (s_1 \cup s_2)}} N_i\right] + p^* \cdot \mathbb{E}\left[ \sum_{\substack{j \in s_1 \cap s_2}} N_j\right] + (p^* + q_{s_1}^*) \cdot \mathbb{E}\left[ \sum_{\substack{l \in s_2 \setminus s_1}} N_l\right]\\
&\phantom{= p^* \cdot \mathbb{E}\left[ \sum_{\substack{i \in V \setminus (s_1 \cup s_2)}} N_i\right] + p^* \cdot \mathbb{E}\left[ \sum_{\substack{j \in s_1 \cap s_2}} N_j\right]} + (p^* - q_{s_1}^*) \cdot \mathbb{E}\left[ \sum_{\substack{r \in s_1 \setminus s_2}} N_r\right] \\
&= U^* + q_{s_1}^* \cdot \left( \mathbb{E}\left[ \sum_{\substack{l \in s_2 \setminus s_1}} N_l\right] - \mathbb{E}\left[ \sum_{\substack{r \in s_1 \setminus s_2}} N_r\right] \right) \\
&> U^*, 
\end{align*}
which contradicts the equilibrium assumption. Therefore, for every pure defense strategy $s_1 \in S$ it holds that $\mathbb{E}\left[\sum_{i \in s_1\phantom{j}} N_i\right] \geq \mathbb{E}\left[\sum_{j \in s_2} N_j\right]$ for every $s_2 \in [\theta]$.

Now we will prove that equilibrium is necessary for the three conditions of the statement to hold. Suppose that all conditions hold and $p^*(G)$ is achieved for the defense strategy $q=(q_1, \dots, q_\theta)$. We will show that the defender and each attacker play a best response. 

Consider an attacker $j \in [k]$ with strategy $t = (t_1, \dots, t_n)$ and support $T_j \subseteq V^*$ according to Condition 2. Her expected payoff is 
\begin{align*}
\sum_{i \in T_j} t_i (1 - p^*) = 1-p^*.
\end{align*}
It suffices to consider unilateral deviations of $j$ to pure strategies. Any pure strategy $i' \in T_j$ gives her expected payoff $1-p^*$, since $p_{i'} = p^*$ (because $T_j \subseteq V^*$). Any pure strategy $i' \in V^*\setminus T_j$ also gives her expected payoff $1 - p^*$ for the same reason. Finally, any pure strategy $i' \in V \setminus V^*$ gives her expected payoff $1 - p_{i'} < 1 - p^*$ by the definition of $V^*$. Therefore every attacker plays a best response.

Now consider the defender with strategy $q = (q_1, \dots, q_\theta)$ and support $S \subseteq [\theta]$. According to Condition 1 of the lemma's statement, $q$ results to vertices of $G$ having vertex probability $p^*$. By Condition 3, for any pure defense strategy $s_1 \in S$ it holds that $\mathbb{E}\left[\sum_{i \in s_1\phantom{j}} N_i\right] \geq \mathbb{E}\left[\sum_{j \in s_2} N_j\right]$ for every $s_2 \in [\theta]$, and let us denote $N_{max} := \mathbb{E}\left[\sum_{i \in s_1\phantom{j}} N_i \right]$. Now consider a unilateral deviation $q' = (q'_1, \dots, q'_\theta)$ of the defender. Her expected payoff is 
\begin{align*}
U(q') &= \sum_{j \in [\theta]} \left( q'_j \mathbb{E}\left[\sum_{i \in j} N_i\right] \right) \\
&\leq \sum_{j \in [\theta]} q'_j N_{max} \\
&= N_{max} \\
&= \sum_{j \in S} \left( q_j \mathbb{E}\left[\sum_{i \in j} N_i\right] \right) \\
&=U(q),
\end{align*}
where the penultimate equation holds due to the fact that $\sum_{j \in S} q_j = 1$.
Therefore, $q$ is a best response for the defender, and the three conditions of the lemma's statement imply a strategy profile that is an equilibrium.
\end{proof}

\begin{lemma}\label{lem:char_DR}
	For any given graph $G$, the equilibrium defense ratio is $\mathrm{DR}(G,S^*) = \frac{1}{p^*(G)}$, where $p^*(G) := \max_{q \in \Delta_{\theta}} \min_{i \in [n]} p_i$ and $S^*$ is an equilibrium.
\end{lemma}

\begin{proof}
	As it is apparent from Theorem~\ref{thm:char_NE}, in an equilibrium, every attacker will have in her support only vertices that are defended with probability exactly $p^*(G)$. Therefore, the expected number of attackers that the defender catches is $p^*(G) \cdot k$. By definition of the defense ratio, $\DR{G}{S^*} = \frac{k}{p^*(G) \cdot k} = \frac{1}{p^*(G)}$.
\end{proof}

\begin{corollary}\label{cor:from_NE_char}
	The following hold:
	\begin{enumerate}[label=(\alph*)]
		\item \label{cor:unique_NE} For a given graph $G$, in any equilibrium, the expected payoff of the defender and each attacker is unique.
		
		\item For a given graph $G$, in any equilibrium with support $S \subseteq [\theta]$ of the defender, for every $s \in S$ there exists a vertex $v \in s$ such that $p_v = p^*(G)$.
%
%
		
		\item \label{cor:eq_k=1} In any CSD game on a graph $G$, the problem of finding the equilibrium defense ratio (or equivalently, $p^*(G)$) for $k \geq 2$ attackers reduces to the same problem in the game with $k = 1$ attacker, which is a two-player constant-sum game. 
	\end{enumerate} 
\end{corollary}

\begin{proof}
	\begin{enumerate}[label=\emph{(\alph*)}]
		\item By Theorem \ref{thm:char_NE}, in an equilibrium the defender chooses a strategy that induces probability $p^*(G)$ to some vertex of $G$ (Condition 1). Also, each of the attackers has in her support $T$ only vertices with vertex probability $p^*(G)$. Therefore, all attackers attack only such vertices and the expected payoff of the defender is $k \cdot p^*(G)$. Consider also an attacker with strategy $t = (t_1, t_2, \dots, t_n)$. Her expected payoff is $\sum_{i \in [n]} t_i (1 - p_{i})$, where $p_i$ is the vertex probability of vertex $i$. This value is equal to $\sum_{i \in T} t_i (1 - p^*(G)) = 1 - p^*(G)$. Since $p^*(G)$ is unique for a graph $G$, the expected payoffs of the defender and each attacker is unique.
		
		\item The proof is by contradiction. Consider an equilibrium where the defender's strategy is $q \in [\theta]$ with support $S$, and there exists a pure strategy $s \in S$ for which every vertex $v \in s$ has $p_v > p^*(G)$. By Condition 2 of Theorem \ref{thm:char_NE}, no attacker has in her support a vertex in $s$. Therefore, the defender can strictly increase her expected payoff by moving all her probability $q_s > 0$ from $s$ to some other pure strategy $s'$ that contains a vertex which is in the support of some attacker.
		
%
%
%
		
		\item Observe that for any given graph $G$, the quantity $p^*(G)$, by definition, only depends on the graph and not the number of attackers $k$. That is, $p^*(G)$ is the same for every $k \geq 1$. Lemma \ref{lem:char_DR} states that in any equilibrium $S^*$, it is $\DR{G}{S^*} = \frac{1}{p^*(G)}$, therefore the defense ratio in an equilibrium does not depend on $k$. This means that when we are given $G$ and we are interested in the equilibrium defense ratio, we might as well consider the game with the single defender and a single attacker. By definition of the game (see Section \ref{sec:prel}) the latter is a two-player constant-sum game. 
	\end{enumerate}
	
\end{proof}

The following corollary implies that coordination (resp. individual selfishness) of the attackers cannot increase the attackers' (resp. defender's) expected payoff in equilibrium.

\begin{corollary}
	\label{cor:coord}
	Every equilibrium with uncoordinated attackers (i.e. as described in Section \ref{sec:prel}) is an equilibrium with coordinated (i.e. centrally controlled) attackers, and vice versa. 
\end{corollary}

\begin{proof}
	Let $q^*$ be a best-defense strategy for the defender. Then, in any best response of any 
	attacker, coordinated or not, every attacker plays only pure strategies that yield maximum
	payoff against $q^*$; i.e. they play only strategies that are defended with probability $p^*(G)$. 
	If this was not the case, either an uncoordinated attacker could increase her payoff by 
	unilaterally changing her strategy, or the ``coordinator'' could increase the payoff the
	attackers collectively get by dictating all the attackers to play vertices that are covered with
	probability $p^*(G)$. 
	
	So, assume that we have an equilibrium in the uncoordinated case. This
	is an equilibrium for the coordinated case as well: according to Theorem \ref{thm:char_NE}, all attackers play vertices that are
	defended with probability $p^*(G)$ and thus the expected collective payoff of the attackers cannot be 
	increased, and furthermore the expected total number of attackers on the vertices of a pure strategy that is in the support of the defender is maximized over all pure defense strategies, so no unilateral deviation of the defender can increase her expected payoff. 
	
	Conversely, in any equilibrium in the coordinated setting the ``coordinator'' dictates 
	all the attackers to attack vertices that are covered with probability $p^*(G)$, satisfying Conditions 1,2 of Theorem \ref{thm:char_NE}. Also in the equilibrium of the coordinated setting, similarly to Condition 3 of Theorem \ref{thm:char_NE}, the ``coordinator'' will have placed the attackers in a way such that the vertices of any pure defense strategy in the support have maximum expected total number of attackers over all pure defense strategies; otherwise the defender can increase her expected payoff by neglecting a pure strategy with smaller than maximum expected total number of attackers, and move the probability assigned on that pure strategy to another one that has maximum expected total number of attackers. By Theorem \ref{thm:char_NE}, this is an equilibrium also for the uncoordinated setting.
\end{proof}

The following theorem provides an algorithm for computing an equilibrium for any CSD game, whose running time is polynomial in $n$ when $\lambda = c$ or $\lambda = n - c$, where $c$ is a constant natural.

\begin{theorem}\label{thm:LP_for_p^*}
	For some given graph $G$ and parameter $\lambda$, there is an algorithm that computes $p^*(G)$ and also finds an equilibrium in time polynomial in $\binom{n}{\lambda}$.
\end{theorem}


\begin{proof}
	Given a graph $G$, the number of attackers $k \geq 1$, and some $\lambda \in \{1,2,\dots, n\}$, the action set $D$ of the defender is constructed by the vertex sets of at most $\binom{n}{\lambda}$ $\lambda$-subgraphs, so
	for $D$'s cardinality $\theta$ it holds that $\theta \leq \binom{n}{\lambda}$. Consider now the mixed strategy $q \in \Delta_{\theta}$ of the defender, where each pure strategy  $j \in [\theta]$ is assigned probability $q_j$. Consider also the vertex probability $p_i$ for each vertex $i \in [n]$. According to Corollary \ref{cor:from_NE_char} \ref{cor:unique_NE} and \ref{cor:eq_k=1}, the unique $p^*(G)$ in the case of a single attacker can be used to derive an equilibrium for the case of $k \geq 2$ attackers. Therefore, we will find $p^*(G)$ for a single attacker, find an equilibrium for that case, and then extend this equilibrium to one in the case of $k \geq 2$ attackers. In more detail, after we find the defense strategy $q^*$ that maximizes $\min_{i \in [n]} p_i$ (Condition 1 of Theorem \ref{thm:char_NE}), i.e. yields $p^*(G)$ on the set $V^* := \arg \max_{q \in \Delta_{\theta}} \min_{i \in [n]} p_i$, an equilibrium is achieved if the single attacker assigns probability $1/|V^*|$ to each vertex of $V^*$; that is because all conditions of Theorem \ref{thm:char_NE} are satisfied. Then, an equilibrium for $k \geq 2$ is achieved if every attacker plays the same strategy as the single attacker; that is because again all conditions of Theorem \ref{thm:char_NE} are satisfied. 
	
	The crucial observation that allows us to design such an algorithm is that we can compute $p^*(G)$ via a Linear Program which has $O\left(\binom{n}{\lambda}\right)$ many variables and $O(n)$ constraints, and therefore its running time is in the worst case polynomial in $\binom{n}{\lambda}$, for $\lambda \in \{ 2,3, \dots, n-1 \}$. For the trivial cases $\lambda=1$ and $\lambda = n$, $D = \{ \{ i \} | i \in V \}$ and $D = V$ respectively, therefore $p^*(G) = 1/n$ and $p^*(G) = 1$ respectively. So in the rest of the proof we will imply that $\lambda \in \{ 2, 3\dots, n-1 \}$. It remains to show how $p^*(G)$ is computed.
%
%
	
	Let us denote $p^*:=p^*(G) := \max_{q \in \Delta_{\theta}} \min_{i \in [n]} p_i$. 
	The computation of $p^*$ can be done as follows: First, consider each of the $\binom{n}{\lambda}$ subsets of $V$ of size $\lambda$, and find if it is a proper $\lambda$-subgraphs of $G$ (i.e. connected); this can be done by running a Depth (or Breadth) First Search algorithm for each subset of size $\lambda$. If it is not, then continue with the next subset. If it is, we consider it in the action set $[\theta]$, and assign to it a variable $q_{j}$ which stands for its assigned probability in a general defense strategy. Now, by definition, for some vertex $i \in [n]$,  $p_i = \sum_{\substack{j \in [\theta]\\i \in j}} q_j$. Therefore, we will consider only pure strategies $j$ which are $\lambda$-subgraphs to create the $p_i$'s. To compute the minimum $p_i$ over all $i$'s we introduce the variable $p'$ and write the following set of $n$ inequalities as a constraint in our Linear Program:
	\begin{align*}
	\sum_{\substack{j \in [\theta]\\i \in j}} q_j \geq p' \quad, \text{ for } i \in \{1, 2, \dots, n\}.
	\end{align*}
	The variable constraints are $p', q_1, q_2, \dots, q_\theta \geq 0$ and also $\sum_{j=1}^{\theta} q_j = 1$, and all of the aforementioned constraints can be written in canonical form by applying standard transformations.
	Finally, the objective function of the Linear Program is variable $p'$ and we require its maximization, which is the value $p^*$. 
\end{proof}

\subsection{Connections to other types of games}

Although CSD games are defined as a normal form game with $k+1$ players, we can observe
that there are equivalent to other well-studied types of games: polymatrix games and 
Stackelberg games.

A polymatrix game is defined by a graph where every vertex represents a player and every 
edge represents a two-player game played by the endpoints of the edge. Every player has
the same set of pure strategies in every game he is involved and to play the game he plays
the same (mixed) strategy in every game. The payoff of every player is the sum they get 
from every two-player game they participate in. In a CSD game we observe the following.
Firstly, the payoff of every attacker depends only on the strategy the defender plays, thus
every attacker is involved only in one two-player game. In addition, all the attackers have the
same set of pure strategies and they share the same payoff matrix.
Similarly, the payoff the defender gets from catching an attacker depends only on the 
strategy the defender and this specific attacker chose. Hence, the payoff of the defender 
can be decomposed into a sum of payoffs from $k$ two-player games. So, a CSD game can be
seen as a polymatrix game where the underlying graph is a star with $k$ leaves that 
correspond to the attackers and the defender is the center of the star. Although many-player
polymatrix games have exponentially smaller representation size compared to the equivalent 
normal-form representation, we should note that this polymatrix game is of exponential size
in the worst case since the defender can have exponential in $n$ pure strategies to choose 
from.

A Stackelberg game is an extensive form two-player game. In the first round, one of the players
commits to a (mixed) strategy. In the second round, the other player chooses a best response
against the committed strategy of her opponent. In a StackeIberg equilirbium the first player is
playing a strategy that maximizes her expected payoff, given that the second player plays a best response (mixed strategy). 
The MaxMin probability $p^*(G)$ for a CSD game on a graph $G$ corresponds to a 
Stackelberg equilibrium. By Corollary \ref{cor:from_NE_char}(c), any CSD game with $k\geq1$ attackers has the same $p^*$ as that of the case with $k=1$. Furthermore, as in a Stackelberg game, in the CSD game with $k=1$ the defender chooses a mixed strategy that maximizes her expected payoff, given that the attacker plays a best response (mixed strategy).
Therefore, when we are interested in the defense-ratio in equilibrium of a CSD game for some arbitrary $k \geq 1$, finding a Stackelberg equilibrium of the corresponding CSD game with $k=1$ suffices.


\section{Defense-Optimal Graphs}

We now focus our attention on defense-optimal graphs. We first characterize defense-optimal
graphs with respect to the MaxMin probability $p^*$ and then use this characterization to 
analyze more specific classes of graphs like cycles and trees. We begin by an exact
computation of the equilibrium defense ratio of any defense-optimal graph.

\begin{theorem}\label{thm:def-opt_ratio}
	In any defense-optimal graph $G$, we have that $\DR{G}{S^*} =\frac{n}{\lambda}$.
\end{theorem}

\begin{proof}
	First we will show that $\frac{n}{\lambda}$ is a lower bound on the Price of Defense and 
	then prove that it is tight. According to Lemma \ref{lem:char_DR}, a lower bound on 
	\PoD$(\lambda)$ can be found by equivalently founding an upper bound on 
	$p^*(G)$
	over all graphs $G$ with $n$ vertices. 
	Let us show that $p^*(G) \leq \frac{\lambda}{ n}$ for every $G$.
	
	Suppose there is a graph $G'$ such that $p^*(G') > \frac{\lambda}{ n}$, and let us focus only on $G'$. Suppose also that the defender has an action set $[\theta]$ on $G'$. Fix the strategy $q=(q_1, \dots, q_{\theta}) \in \Delta_{\theta}$ that achieves $p^*(G')$. Then, by definition of $p^*(G')$, for the vertex probabilities $p_i$ it holds that $p_i > \frac{\lambda}{ n}$ for all $i \in [n]$. Therefore, it is
	\begin{align}\label{eq:sumP>lambda}
	\sum_{i=1}^{n} p_i > \lambda.
	\end{align}
	Also, by definition of a defense strategy, if $X$ denotes the random variable corresponding to the number of vertices that the defender covers, then:
	\begin{equation}\label{eq:num_of_cov_vert1}
	\mathbb{E}[X] = \sum_{j=1}^{\theta} q_{j} \cdot |L_j| = \lambda \qquad 
	(\text{ where $L_j$ is a $\lambda$-subgraph of $G$, hence $|L_j|=\lambda \quad \forall j \in [\theta]$}).
	\end{equation}

	Let us introduce the indicator variables $X_{ij}$, $i \in [n]$, $j \in [\theta]$ with value 1 if vertex $i \in L_j$, and 0 otherwise. Then,
	\begin{align}\label{eq:num_of_cov_vert2}
	\mathbb{E}[X] &= \sum_{j=1}^{\theta} q_{j} \sum_{i=1}^{n} X_{ij} \nonumber  \\
	&= \sum_{i=1}^{n} \sum_{j=1}^{\theta} q_{j} X_{ij} \nonumber  \\
	&= \sum_{i=1}^{n} p_{i} \\
	&> \lambda   \qquad (\text{by inequality } \eqref{eq:sumP>lambda}),  \nonumber  
	\end{align}
	which contradicts \eqref{eq:num_of_cov_vert1}.
	
	It remains to show that the lower bound $\frac{n}{\lambda}$ on the \PoD$(\lambda)$ is tight. This is easy to do by showing that $\frac{\lambda}{ n}$ is a tight upper bound on $p^*(G)$: observe that every vertex of the line graph with $n= r \cdot \lambda$ vertices, where $r \in \mathbb{N}^\ast$, can be covered with $\frac{n}{\lambda}$ disjoint pure strategies of the defender. Therefore, the defender can assign probability $1/(n / \lambda)$ to each pure strategy, and in that case, $p^*(G) = \frac{\lambda}{n}$. 
\end{proof}

As an intermediate corollary of  Theorem~\ref{thm:def-opt_ratio} we get the following 
characterisation of defense-optimal graphs.

\begin{corollary}\label{cor:def_opt}
	A graph $G$ is defense-optimal if and only if all of its vertices are defended with probability $\frac{\lambda}{n}$.
\end{corollary}

\begin{proof}
	Necessity of defense-optimality is trivial: every vertex has vertex probability $\frac{\lambda}{n}$, therefore $p^*(G) = \frac{\lambda}{n}$, so by Theorem \ref{thm:def-opt_ratio} the graph is defense-optimal.
	
	Sufficiency of defense-optimality is also easy to see using the equations \eqref{eq:num_of_cov_vert1}, \eqref{eq:num_of_cov_vert2} of the proof of Theorem \ref{thm:def-opt_ratio}. Suppose that the graph is defense-optimal and consider an equilibrium where the defense strategy is $q = (q_1, \dots, q_{\theta})$. Then the sum of vertex probabilities is $\sum_{i=1}^{n} p_{i} = \lambda$ according to the aforementioned equations. Therefore, if there exists a vertex $v$ with vertex probability $p_v > \frac{\lambda}{n}$ then there is another vertex $u$ with probability $p_u < \frac{\lambda}{n}$. This means that $p^*(G) <\frac{\lambda}{n}$, and as a result the graph is not defense-optimal which contradicts our assumption.
\end{proof}

Someone may wonder whether Corollary~\ref{cor:def_opt} can be further exploited to prove 
that, in general, 	best-defense strategies in defense-optimal graphs are uniform, i.e. every 
pure strategy $s$ in the support $S$ of the defender is assigned probability $1/|S|$. However, as we demonstrate in Figure~\ref{graph1} this is not the case. On the other hand, this claim is
true for cyclic graphs and trees.

\begin{figure}[h]
	\begin{center}
		\includegraphics[scale=0.42]{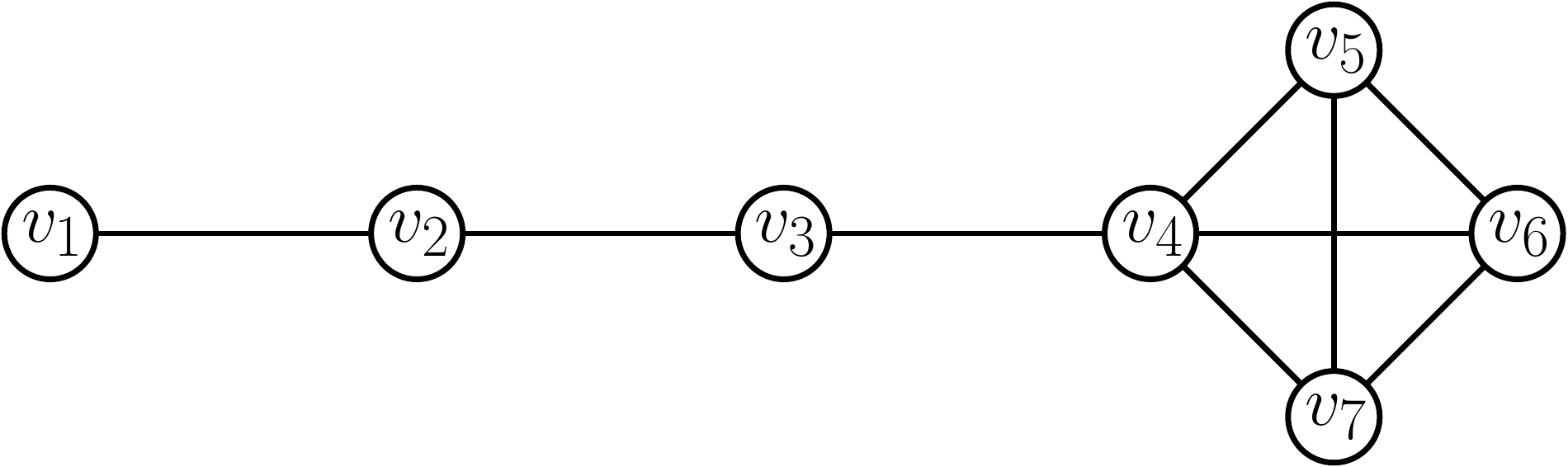}
	\end{center}
	\caption{Here $n=7$, $\lambda = 3$ and $p^*(G) = 3/7$ is achievable by assigning probability $3/7$ to pure strategy $\{v_1, v_2, v_3\}$ and probability $1/7$ to each of the pure strategies $\{v_4, v_5, v_6\}$, $\{v_4, v_5, v_7\}$, $\{v_4, v_6, v_7\}$, $\{v_5, v_6, v_7\}$, so the graph is defense optimal. However, observe that $v_1$ cannot participate in more than one pure strategies, so in a uniform defense strategy with support of size $r$, the vertex probability $p_{v_1}$ has to be $1/r$ (by definition of uniformity), but it also has to be $3/7$. Since $r \in \mathbb{N}$, this is a contradiction.}\label{graph1}
\end{figure}

\begin{observation}
	All cyclic graphs are defense-optimal.
\end{observation}

\begin{proof}
	Consider an arbitrary cyclic graph $G$ with $n$ vertices. We will show that the graph can achieve vertex probability $p_{i}=\frac{\lambda}{n}$ for every $i \in [n]$, thus by Corollary~\ref{cor:def_opt} it is defense-optimal. Consider the whole action set $D$ of the defender, i.e. every path starting from a vertex $i$ going clockwise and ending at vertex $i+\lambda - 1$. Observe that there are only $n$ such paths, therefore $\theta:=|D|=n$. By assigning probability $\frac{1}{n}$ to each pure strategy $j \in [\theta]$, since each vertex is in exactly $\lambda$ pure strategies, each vertex $i \in [n]$ has vertex probability $p_i = \lambda \cdot \frac{1}{\theta} = \frac{\lambda}{n}$.
\end{proof}

\subsection{Tree Graphs}\label{sec:tree_graphs}
In this section we focus on the case where the graph is a tree. We first further refine the 
characterization of defense-optimal graphs for trees. Then, we utilise this characterisation to 
derive a polynomial-time algorithm that decides in polynomial time whether a given tree is defense-optimal, and if that is the case, it constructs in polynomial time a defense-optimal strategy for it. On the other hand, in the case where the tree is not defense-optimal, we show that it is \NP-hard to compute a best-defense strategy for it, namely it is \NP-hard to compute $p^*(G)$. We first provide Lemma~\ref{lem:tree_opt} which will be used in our polynomial-time algorithm for checking defense-optimality on trees. Henceforth, we write that a graph is {\em covered} by a defense strategy if every vertex of the graph is covered by a $\lambda$-subgraph that is in the support of the defense strategy.

\begin{lemma}\label{lem:tree_opt}
	A tree $T$ is defense-optimal if and only if $T$ can be decomposed into $\frac{n}{\lambda}$ disjoint $\lambda$-subgraphs.
\end{lemma}

\begin{proof}
		$\pmb{(\Rightarrow)}$ Let $T$ be defense-optimal. We will show that the support of any best defense strategy on $T$ must comprise of pure strategies that are disjoint $\lambda$-subgraphs which altogether cover every $v\in V$. Since those are disjoint and cover $T$, it follows that their number is $\frac{n}{\lambda}$ in total.

		If $\lambda=1$ then the above trivially holds. Assume that $\lambda \geq 2$ and consider a best defense strategy on $T$ whose support comprises of a collection $\mathcal{L}$ of $\lambda$-subgraphs.
		
		Let $u \in V$ be a leaf of $T$ and let $v \in V$ be its parent. Any $\lambda$-subgraph  in $\mathcal{L}$ covering $u$ must also cover $v$, since $\lambda\geq 2$. Also, any $\lambda$-subgraph in $\mathcal{L}$ covering $v$ must also cover $u$, otherwise $p_v$ would be greater than $p_u$. Now, consider the neighbors of $v$. For those of them that are leaves, the same must hold as holds for $u$, namely $v$ and its leaf-children must all be covered by the same exact $\lambda$-subgraph(s).
		
		Consider the case where there is a leaf $u \in V$, such that a \emph{single} $\lambda$-subgraph contains $u$, its parent $v$, and all the other leaf-children of $v$ (and, possibly, other vertices connected to $v$). Then we can remove this $\lambda$-subgraph from $\mathcal{L}$ and the corresponding tree from $T$. This leaves the remainder of $T$ being a forest comprising of trees $T_1, \ldots, T_x$, each of which has a (best) defense strategy comprising of the corresponding subset of (the remainder of) $\mathcal{L}$ on $T_i$. Notice that it must be the case that every tree $T_i$, $i=1,2,\ldots,x$, has size at least $\lambda$ (otherwise the initial collection $\mathcal{L}$ would not have covered $T$). So, if there is always a leaf $u$ in some tree of the forest, such that a \emph{single} $\lambda$-subgraph contains $u$, its parent $v$, and all the other leaf-children of $v$ (and, possibly, other vertices connected to $v$), we can proceed in the same fashion for each of the $T_i$'s, always removing a $\lambda$-subgraph from $\mathcal{L}$, and the corresponding vertices from $T$, until we end up with an empty tree. This means that $\mathcal{L}$ was indeed a collection of disjoint $\lambda$-subgraphs covering $T$.
		
		However, assume for the sake of contradiction that at some ``iteration'' the assumption does not hold, namely assume that there is a tree in the forest with no leaf $u$, such that a single $\lambda$-subgraph contains $u$, its parent $v$, and all the other leaf-children of $v$ (and, possibly, other vertices connected to $v$). This means that there are (at least) two $\lambda$-subgraphs in $\mathcal{L}$, namely $L_1, L_2$, that cover $u$. Due to our initial observations, $u$, together with its parent $v$ and all of $v$'s leaf-children are contained in both $L_1$ and $L_2$. Since those are different $\lambda$-subgraphs, there is a vertex $z$ in the tree which belongs to $L_2$ but does not belong to $L_1$. Since $p_z = p_v$ (due to the fact that $\mathcal{L}$ is the support of the defense-optimal strategy and Corollary \ref{cor:def_opt}), it must hold that there is a different $\lambda$-subgraph, $L_3$, which covers $z$ but does not cover $v$ or any of its leaf-children. If $L_3$ also covers a vertex in $L_1 \setminus L_2$\footnote{We use $L_i \setminus L_j$ for some $\lambda$-subgraphs $L_i, L_j$ to denote the set of vertices which are contained in $L_i$ but not in $L_j$.}, then there is a cycle in the tree which is a contradiction. So $L_3$ must not cover vertices in $L_1 \setminus L_2$. Since $\L_3$ is different to $L_2$, there must be a vertex $z'$ in the tree which belongs in $L_3$ but not in $L_2$ (also not in $L_1$). Since $p_{z'} = p_{z}$ (due to the fact that $\mathcal{L}$ is the support of the defense-optimal strategy and Corollary \ref{cor:def_opt}), it must hold that there is a different $\lambda$-subgraph, $L_4$, which covers $z'$ but does not cover $z$ or any of the vertices in $L_2$. Similarly to before, if $L_4$ covers a vertex in $L_1 \setminus L_2$, then there is a cycle in the tree which is a contradiction. So $L_4$ must not cover vertices in $L_1$ or in $L_2$.
		
		Proceeding in the same way, we result in contradiction since the tree has finite number of vertices and there will need to be an overlap in coverage of some $L_j$ with some $L_i$, $j> i+1$, which would mean that there is a cycle in the tree.
		
		Therefore, there cannot be any overlaps between the $\lambda$-subgraphs of $\mathcal{L}$, meaning that $\mathcal{L}$ comprises of $\frac{n}{\lambda}$ disjoint $\lambda$-subgraphs which altogether cover $T$.
		
		$(\pmb{\Leftarrow})$ Let $\mathcal{L} = \{L_1, \ldots, L_{\frac{n}{\lambda}} \}$ be a collection of $\frac{n}{\lambda}$ disjoint $\lambda$-subgraphs that altogether cover $T$. Let the defender play each $L_i$, $i\in \{1,\ldots,\frac{n}{\lambda} \}$, equiprobably, that is, with probability $1/ \left( \frac{n}{\lambda} \right) = \frac{\lambda}{n}$. Then every vertex $v \in V$ is covered with probability $p_v = \frac{\lambda}{n} = p^*(G)$, meaning that $T$ is defense-optimal. 
\end{proof}

With Lemma~\ref{lem:tree_opt} in hand we can derive a polynomial-time algorithm that 
decides if a tree is defense-optimal, and if it is, to produce a best-defense strategy.
\begin{theorem}
\label{thm:trees-algo}
There exists a polynomial-time algorithm that decides whether a tree is defense-optimal and 
produces a best-defense strategy for it, or it outputs that the tree is not defense-optimal.
\end{theorem}
\begin{proof}
The algorithm works as follows. Initially, there is a pointer associated with a counter  
in every leaf of the tree $T$ that moves ``upwards'' towards an arbitrary root of the tree. For 
every move of the pointer the  corresponding counter increases by one. 
The pointer moves until one of the following happens: either the counter is equal to 
$\lambda$, or it reaches a vertex with degree greater of equal to 3 where it ``stalls''.
In the case where the counter is equal to $\lambda$, we create a $\lambda$-subgraph of
$T$, we delete this $\lambda$-subgraph from the tree, we move the pointer one position 
upwards, and we reset the counter back to zero.  
If a pointer stalls at a vertex of degree $d \geq 3$, it waits until all $d-1$ pointers reach 
this vertex. Then, all these pointers are merged to a single one and a new counter is created 
whose value is equal to the sum of the counters of all $d$ pointers.
If this sum is more than $\lambda$, then the algorithm returns that the graph is not 
defense-optimal.  If this sum is less than or equal to $\lambda$, then we proceed as if there 
was initially only this pointer with its counter; if the new counter is equal to $\lambda$, 
then we create a $\lambda$-subgraph of $T$ and reset the counter to 0; else the pointer 
moves upwards and the counter increases by one. To see why the algorithm requires 
polynomial time, observe that we need at most $n$ pointers and $n$ counters and in addition
every pointer moves at most $n$ times.

We now argue about the correctness of the algorithm described above. Clearly, if the algorithm 
does not output that the tree is not defense-optimal, it means that it partitioned $T$ into 
$\lambda$-subgraphs. So, from Lemma~\ref{lem:tree_opt} we get that $T$ is 
defense-optimal and the uniform probability distribution over the produced partition covers
every vertex with probability $\frac{\lambda}{n}$. It remains to argue that when the algorithm 
outputs that the graph is not defense-optimal, this is indeed the case. Consider the
case where we delete a $\lambda$-subgraph of the (remaining) tree.
Observe that the $\lambda$-subgraph our algorithm created deleted should be uniquely 
covered by this $\lambda$-subgraph in any best-defense strategy; any other 
$\lambda$-subgraph would overlap with some other $\lambda$-subgraph. 
Hence, the deletion of such a $\lambda$-subgraph was not a ``wrong''
move of our algorithm and the remaining tree is defense-optimal if and only if the tree 
before the deletion was defense-optimal. This means that any deletion that occurred 
by our algorithm did not make the remaining graph non defense-optimal.
So, consider the case where after a merge that occurred at vertex $v$ we get that the 
new counter is $c > \lambda$. Then, we can deduce that all the subtrees rooted at $v$ 
associated with the counters have strictly less than  $\lambda$ vertices. Hence, in order to 
cover all the $c > \lambda$ vertices using $\lambda$-subgraphs, at least two of these 
$\lambda$-subgraphs cover vertex $v$. 
Hence, the condition of Lemma~\ref{lem:tree_opt} is violated. But since every step of our algorithm so far was 
correct, it means that $v$ cannot be covered only by one $\lambda$-subgraph. Hence, our
algorithm correctly outputs that the tree is not defense-optimal.
\end{proof}


In Theorem~\ref{thm:trees-algo} we showed that it is easy to decide whether a tree is
defense-optimal and if this is the case, it is easy to find a best-defense strategy for it.
Now we prove that if a tree is not defense-optimal, then it is \NP-hard to find 
a best-defense strategy for it.
\begin{theorem}
	\label{thm:np-hard}
	Finding a best-defense strategy in \csd games is \NP-hard, even if the graph is 
	a tree.
\end{theorem}
\begin{proof}
	We will prove the theorem by reducing from \trpartition. In an instance of \trpartition we
	are given a multiset with $n$ positive integers $a_1, a_2, \ldots, a_n$ where $n=3m$ for some $m \in \mathbb{N}_{>0}$
	and we ask whether it can be partitioned into $m$ triplets $S_1, S_2, \ldots, S_m$ such 
	that the sum of the numbers in each subset is equal. Let $s = \sum_{i=1}^n a_i$. Observe
	then that the problem is equivalent to asking whether there is a partition of the integers to 
	$m$ triplets such that the numbers in every triplet sum up to $\frac{s}{m}$.
	Without loss of generality we can assume that $a_i < \frac{s}{m}$ for every $i \in [n]$; 
	if this was not the case, the problem could be trivially answered.
	So, given an instance of \trpartition, we create a tree $G=(V,E)$ with $s+1$
	vertices and $\lambda = \frac{s}{m}+1$. 
	The tree is created as follows. For every integer $a_i$, we create a path with $a_i$
	vertices. In addition, we create the vertex $v_0$ and connect it to one of the two 
	ends of each path. We will ask whether $p^*(G) \geq \frac{1}{m}$.
	
	Firstly, assume that the given instance of \trpartition is satisfiable. Then, given $S_j$ we create
	a $(\frac{s}{m}+1)$-subgraph of $G$ as follows. If $a_i \in S_j$, then we add the 
	corresponding path of $G$ to the subgraph. Finally, we add vertex $v_0$ in our $(\frac{s}{m}+1)$-subgraph and the resulting subgraph is connected (by the construction of $G$).
	Since the sum of $a_i$'s equals $\frac{s}{m}$, the constructed subgraph has 
	$\frac{s}{m}+1$	vertices.
	If we assign probability $\frac{1}{m}$ to every $(\frac{s}{m}+1)$-subgraph we get that 
	$p_v \geq \frac{1}{m}$ for every $v \in V$.
	
	To prove the other direction, assume that $p^*(G) \geq \frac{1}{m}$ and observe 
	the following. 
	Firstly, since as we argued it is $a_i < \frac{s}{m}$ for every $i \in [n]$, it holds 
	that every $(\frac{s}{m}+1)$-subgraph of $G$ contains vertex $v_0$. 
	Thus, $p_{v_0}=1$ and $\sum_{v \neq v_0} p_v \geq \frac{s}{m}$, since there are $s$
	vertices other than $v_0$ and for each one of them holds that $p_v \geq \frac{1}{m}$.
	In addition, observe that $\sum_{v \in V} p_v = \lambda = \frac{s}{m}+1$. 
	Hence, we get that $p_v = p^*(G) = \frac{1}{m}$ for every vertex $v \neq v_0$. 
	In addition, observe that every pure defense strategy that covers a leaf of this tree, 
	covers all the vertices of the branch. Hence, for every branch of the tree, all its
	vertices are covered by the same set of pure strategies; if a vertex $u$ that is closer
	to $v_0$ is covered by one strategy that does not cover the whole branch, then the 
	leaf $u'$ of the branch is covered with probability less than $u$. 
	So, in order for $p_v = p^*(G) = \frac{1}{m}$ for every $v \neq v_0$, it means that there 
	exist a $(\frac{s}{m}+1)$-subgraph that {\em exactly} covers a subset of the paths;
	this means that	if a $(\frac{s}{m}+1)$-subgraph covers a vertex in a path, then 
	it covers every vertex of the path. 
	Hence, by the construction of the graph, we get that this $(\frac{s}{m}+1)$-subgraph 
	of $G$ corresponds to a subset of integers in the \trpartition instance that sum up to 
	$\frac{s}{m}$. Since, \trpartition is \NP-hard, we get that finding a best defense strategy is 
	\NP-hard.
\end{proof}

\subsection{General Graphs}\label{sec:general_graphs}

We conjecture that contrary to checking defense-optimality of tree graphs and constructing a corresponding defense-optimal strategy in polynomial time, it is \NP-hard to even decide whether a given (general) graph is defense-optimal.

\begin{conjecture}
	\label{thm:np-general}
	It is \NP-hard to decide whether a graph is defense-optimal.
\end{conjecture}

\section{Approximation algorithm for $p^*(G)$}
We showed in the previous section that, given a graph $G$, it is \NP-hard to find the best-defense strategy, or equivalently, to compute $p^*(G)$. We also presented in Theorem \ref{thm:LP_for_p^*} an algorithm for computing the exact value $p^*(G)$ of a given graph $G$ (and therefore its best defense ratio), but this algorithm has running time polynomial in the size of the input only in the cases $\lambda = c$ or $\lambda = n - c$, where $c$ is a constant natural.
On the positive side, we present now a polynomial-time algorithm which, given a graph $G$ of $n$ vertices, returns a defense strategy with defense ratio which is within factor $2+\frac{\lambda - 3}{n}$ of the best defense ratio for $G$. In particular, it achieves defense ratio $1 / p' \leq \left( 2+\frac{\lambda - 3}{n} \right) / p^*(G) $, where $p' = \min_{i \in [n]} p_i$ and every $p_{i}$, $i\in [n]$ is the vertex probability determined by the constructed defense strategy.
We henceforth write that a collection $\mathcal{L}$ of $\lambda$-subgraphs covers a graph $G=(V,E)$, if every vertex of $V$ is covered by some $\lambda$-subgraph in $\mathcal{L}$. 
The algorithm presented in this section returns a collection $\mathcal{L}$ of at most $\frac{2n-3}{\lambda}+1$  $\lambda$-subgraphs that covers $G$. Therefore, the uniform defense strategy over $\mathcal{L}$ assigns probability at least $1 / \left( \frac{2n-3}{\lambda}+1 \right)$ to each $\lambda$-subgraph.

For any collection $\mathcal{L}$ of $\lambda$-subgraphs and for any $v \in V$, let us denote by $\coverage_\mathcal{L}(v)$ the number of $\lambda$-subgraphs in $\mathcal{L}$ which $v$ belongs in. Observe that:
\begin{align}\label{eq:coverage}
\sum_{v \in V} \coverage_\mathcal{L}(v) = |\mathcal{L}|\cdot \lambda,
\end{align}
where $|\mathcal{L}|$  denotes the cardinality of $\mathcal{L}$. 

We first prove Lemma~\ref{lem:coverage}, to be used in the proof of the main theorem of this Section.
We henceforth denote by $V(G)$ and $E(G)$ the vertex set and edge set, respectively, of some graph $G$. 

\begin{lemma}\label{lem:coverage}
	For any tree $T$ of $n$ vertices, and for any $\lambda \leq n$, we can find a collection $\mathcal{L}$ of distinct $\lambda$-subgraphs such that for every $v \in V$, it holds that $1 \leq \coverage_\mathcal{L} (v) \leq \text{degree}(v)$, except maybe for (at most) $\lambda-1$ vertices, where for each  of them it holds that $\coverage_\mathcal{L}(v) = \text{degree}(v)+ 1$.
\end{lemma}

\begin{proof}
	We will prove the statement of the lemma by providing Algorithm \ref{alg:main} that takes as input $T$ and $\lambda$ and outputs the requested collection $\mathcal{L}$ of $\lambda$-subgraphs.

	
	\begin{algorithm}[!h]
		\caption{\textsc{Main Algorithm}}\label{alg:main}
		\begin{algorithmic}[1]
			\REQUIRE{A tree graph $T=(V,E)$ of $n$ vertices, and a natural $\lambda \leq n$.}
			\ENSURE{A collection $\mathcal{L}$ of distinct $\lambda$-subgraphs that satisfies the statement of Lemma \ref{lem:coverage}.}
			
			\medskip
			
			\STATE{$i$, \quad global variable.} \COMMENT{\quad \% The index of the $\lambda$-subgraph $L_i$.}
			\STATE{$count$, \quad global variable.} \COMMENT{\quad \% Is 0 until the whole tree is covered, then it becomes 1 to allow for the last $\lambda$-subgraph to be completed, if it is not already.}
			\STATE{$S$, \quad global variable.} \COMMENT{\quad \% The set of vertices already covered by the algorithm.}
			\STATE{$vertex$, \quad global variable.} \COMMENT{\quad \% The vertex considered to be inserted in a $\lambda$-subgraph. }
			
			~\\
			
			\medskip
			
			\STATE{$S \gets \emptyset$}
			\STATE{$i \gets 1$}
			\STATE{$L_i \gets \emptyset$}
			\STATE{Pick an arbitrary vertex $v$ of $T$ and consider it the root.}
			\STATE{$vertex \gets v$}
			\STATE{$count \gets 0$}
			
			\medskip
			
			\WHILE{$count < 2$}
			
				\WHILE{$S \neq V$}\label{ln:check_if_covered-start}

					\WHILE [\quad \% The while-loop to ensure that the first element of $L_i$ is uncovered.] {$vertex \in S$}\label{ln:find_unc_first_node}
						\IF{$vertex$ has a child $u \notin S$}\label{ln:DFS_1_start}
							\STATE{$vertex \gets u$}\label{ln:if_uncov_child_1}
						\ELSE\label{ln:if_cov_child_1}
							\STATE{$vertex \gets$ parent of $vertex$}\label{ln:DFS_1_end}
						\ENDIF\label{ln:search_for_unc_ver-end}
					\ENDWHILE
					
					\WHILE [ \quad \% The while-loop that fills in the current $\lambda$-subgraph $L_i$.] {$|L_i| < \lambda$}\label{ln:fill_in-start}
						\STATE{$L_i \gets L_i \cup \{vertex\}$}\label{ln:insert_to_subgraph}
						\STATE{$S \gets S \cup \{vertex\}$}\label{ln:insert_to_cov_set}
						\IF{$vertex$ has a child $u \notin S$}\label{ln:DFS_2_start}
							\STATE{$vertex \gets u$}\label{ln:if_uncov_child_2}
						\ELSE\label{ln:if_cov_child_2}
							\STATE{$vertex \gets$ parent of $vertex$}\label{ln:DFS_2_end}
						\ENDIF
					\ENDWHILE\label{ln:fill_in-end}
						\IF{$count < 1$}
							\STATE{$i \gets i+1$}\label{ln:new_subgraph}
							\STATE{$L_i \gets \emptyset$}\label{ln:initialize_subgraph}
						\ELSE\label{ln:no_new_subgraph_1}
							\BREAK\label{ln:no_new_subgraph_2}
						\ENDIF
					\ENDWHILE\label{ln:check_if_covered-end}
				
				\STATE{$S \gets \emptyset$}\label{ln:last_DFS_1}
				\STATE{$i \gets i-1$}
				\STATE{Pick an arbitrary vertex $v \in L_i$ and consider it the root.}
				\STATE{$vertex \gets v$}\label{ln:last_DFS_3}
				\STATE{$count \gets count + 1$}\label{ln:last_DFS_4}
				%
				
			\ENDWHILE
			
		\end{algorithmic}
	\end{algorithm}

	The algorithm starts by picking an arbitrary vertex $v$ to serve as the root of the tree. Then it performs a Depth-First-Search (DFS) starting from $v$. We will distinguish between {\em visiting} a vertex and {\em covering} a vertex in the following way. We say that DFS visited a vertex if it considered that vertex as a candidate to be inserted to some $\lambda$-subgraph, and we say that DFS covered a vertex if it visited {\em and} inserted the vertex at some $\lambda$-subgraph. By definition, DFS visits in a greedy manner first an uncovered child, and only if there is no such child, it visits its parent (lines \ref{ln:DFS_1_start}-\ref{ln:DFS_1_end}, \ref{ln:DFS_2_start}-\ref{ln:DFS_2_end}). The set-variable that keeps track of the covered vertices is $S$.
	
	Starting with the root of $T$, the algorithm simply visits the whole vertex set according to DFS, putting each visited vertex in the same $\lambda$-subgraph $L_i$ (starting with $i=1$) (lines \ref{ln:fill_in-start}-\ref{ln:fill_in-end}), and when $|L_i| = \lambda$, a new empty $\lambda$-subgraph $L_{i+1}$ is picked to get filled in with $\lambda$ vertices (lines \ref{ln:new_subgraph}-\ref{ln:initialize_subgraph}) taking care of one extra thing:
	The first vertex that the algorithm puts in an empty $\lambda$-subgraph $L_i$, $i \in \{1,2, \dots\}$ is guaranteed to be one that has not been covered by any other $\lambda$-subgraph so far (lines \ref{ln:find_unc_first_node}-\ref{ln:search_for_unc_ver-end}). This ensures that no two $\lambda$-subgraphs will eventually be identical.
	
	The algorithm will not only visit all vertices in $T$, but also cover them. That is because there is no point where the algorithm  checks whether the currently visited vertex is uncovered and then does not cover it. On the contrary, it covers every vertex that it visits, except for some already covered one in case the current $\lambda$-subgraph is empty (lines \ref{ln:find_unc_first_node}-\ref{ln:DFS_2_end}). And since DFS by construction visits every vertex, we know that at some point the whole vertex set will be covered, or equivalently, $\coverage_{\mathcal{L}}(v) \geq 1, \forall v \in V$. Therefore, the algorithm will eventually exit the while-loop in lines \ref{ln:check_if_covered-start}-\ref{ln:check_if_covered-end}.
	
	Now we prove that, after the algorithm terminates, every vertex $v \in V$ is covered at most $degree(v)$ times, except for at most $\lambda - 1$ vertices that are covered $degree(v)+1$ times. Observe that DFS visits every vertex $v$ at most $degree(v)$ times: (a) $v$ will be visited after its parent $u$ only if $v$ is uncovered (lines \ref{ln:DFS_1_start}-\ref{ln:if_uncov_child_1}, \ref{ln:DFS_2_start}-\ref{ln:if_uncov_child_2}), $v$ will get covered (lines \ref{ln:insert_to_subgraph}-\ref{ln:insert_to_cov_set}), and will not get visited ever again by its parent since it will be covered (lines \ref{ln:if_cov_child_1}-\ref{ln:DFS_1_end}, \ref{ln:if_cov_child_2}-\ref{ln:DFS_2_end}). (b) $v$ will be visited at most once by each of its children, say $w$, only if $w$ does not have an uncovered child (lines \ref{ln:if_cov_child_1}-\ref{ln:DFS_1_end}, \ref{ln:if_cov_child_2}-\ref{ln:DFS_2_end}), and $v$ will not get ever visited by its parent since $v$ will be covered, and also $v$ cannot be visited a second time by any of its children, since they can never be visited again (they can only be visited through $v$ since $T$ is a tree). Therefore, any vertex $v$ will be visited exactly once after its parent is visited, and at most once by each of its children, having a total of at most $degree(v)$ visits. And since, as argued above, the total number of visits of a vertex is at most the number of times it will be covered, when DFS terminates, that is $S=V$, it will be $\coverage_{\mathcal{L}}(v) \leq degree(v)$, for every $v \in V$. 
	
	However, note that the last nonempty $\lambda$-subgraph $L_i$ might not consist of $\lambda$ vertices since the entire $V$ was covered and DFS could not proceed further. In this case, the algorithm empties the set $S$ that keeps track of the covered nodes, takes the current $L_i$ and fills it in with exactly another $\lambda - |L_i|$ vertices. This is done by picking an arbitrary vertex from $L_i$ and setting it as the root of $T$, and performing one last DFS starting from it until $L_i$ has $\lambda$ vertices in total (lines \ref{ln:last_DFS_1}-\ref{ln:last_DFS_3}). To ensure that the DFS will continue only until it fills in this current $L_i$, the algorithm counts the number of times that it runs the while-loop of DFS, namely lines \ref{ln:check_if_covered-start}-\ref{ln:check_if_covered-end}, via the variable $count$ (line \ref{ln:last_DFS_4}), which escapes the while-loop of DFS in case DFS has filled in $L_i$ (lines \ref{ln:no_new_subgraph_1}-\ref{ln:no_new_subgraph_2}) and terminates. Observe that in the last $\lambda$-subgraph $L_i$, a vertex $v$ inserted in the last iteration of DFS ($count=1$) and was not inserted in $L_i$ by the first run ($count=0$) might have been covered by the first run of DFS exactly $degree(v)$ times, therefore when the algorithm terminates it has been covered $degree(v)+1$ times. Since by the end of the first DFS run $L_i$ had at least one vertex, the cardinality of such vertices that are covered more times than their degree are at most $\lambda - 1$.	
\end{proof}

We can now prove the following.
\begin{lemma}\label{lem:num_of_subgraphs}
	For any graph $G$ of $n$ vertices, and for any $\lambda \leq n$, there exist (at most) $\frac{2n - 3}{\lambda} +1$ $\lambda$-subgaphs of $G$ that cover $G$.
\end{lemma}
%
\begin{proof}
	Consider a spanning tree $T$ of $G$. Then Lemma \ref{lem:coverage} applies to $T$. Observe that a collection $\mathcal{L}$ as described in the statement of the aforementioned lemma has the same qualities for $G$ since $V(T) = V(G)$ and $E(T) \subseteq E(G)$. That is, $\mathcal{L}$ is a collection of distinct $\lambda$-subgraphs of $G$, such that for every $v \in V$, it holds that $1 \leq \coverage_\mathcal{L} (v) \leq \text{degree}(v)$, except maybe for (at most) $\lambda-1$ vertices, for each $v$ of which it is $\coverage_\mathcal{L}(v) = \text{degree}(v)+ 1$.
	
	Fix a particular value for $\lambda$ and consider a collection $\mathcal{L}$ of $\lambda$-subgraphs as constructed in the proof of Lemma~\ref{lem:coverage}. Then, by equation \eqref{eq:coverage},
	\[|\mathcal{L}|  = \frac{\sum_{v \in V} \coverage_\mathcal{L} (v) }{\lambda}   \leq \frac{\sum_{v \in V}  \text{degree}(v)  +   (\lambda-1) }{\lambda}   =   \frac{2(n-1)}{\lambda} + \frac{\lambda-1}{\lambda}  \leq \frac{2n - 3}{\lambda} +1. \] 
\end{proof}

We conclude with the simple algorithm that achieves a defense strategy with
defense ratio which is within factor $2+\frac{\lambda - 3}{n}$ of the best defense ratio for G.

\begin{algorithm}[!h]
	\caption{\textsc{Approximating the best defense ratio}}  
	\label{alg:approximation_of_p*}  
	\begin{algorithmic}[1]
		\REQUIRE{Graph $G=(V,E)$ of $n$ vertices, a natural $\lambda \leq n$.}
		\ENSURE{A defense strategy that satisfies the statement of Theorem \ref{thm:appr_def_rat}.}
		
		\medskip
		
		\STATE{Find a spanning tree $T$ of $G$.}\label{step:span}
		\STATE{Construct a collection $\mathcal{L}$ of $\lambda$-subgraphs of $T$ as described in the proof of Lemma~\ref{lem:coverage}.}\label{step:collection}
		
		\STATE{Assign probability $q_i = \frac{1}{|\mathcal{L}|}$ to every $\lambda$-subgraph in $\mathcal{L}$, $i=1,2,\ldots, |\mathcal{L}|$.} \label{step:prob}
		
		\medskip
		
		\RETURN{The above uniform defense strategy over the collection $\mathcal{L}$.}
	\end{algorithmic}
\end{algorithm}

\begin{theorem}\label{thm:appr_def_rat}
	Given any graph $G=(V,E)$, Algorithm~\ref{alg:approximation_of_p*} computes in time $O(|E|)$ a defense strategy such that, for any combination of attack strategies, the resulting strategy profile $S$ yields defense ratio $\mathrm{DR}(G,S) \leq \left( 2+\frac{\lambda - 3}{n} \right) \cdot \mathrm{DR}(G,S^*)$.
\end{theorem}

\begin{proof}
	As argued in Lemma \ref{lem:num_of_subgraphs}, there is a collection $\mathcal{L}$ of $\lambda$-subgraphs with $|\mathcal{L}| \leq \frac{2n}{\lambda} + 1 - \frac{3}{\lambda}$ which covers $G$.
	Therefore, given the uniform defense strategy returned by Algorithm~\ref{alg:approximation_of_p*} (which determines the vertex probability $p_i$ for each vertex $i$) achieves a minimum vertex probability $p' := \min_{i \in [n]} p_i$ for which it holds that:
	\begin{align*}
	p' = \frac{1}{|\mathcal{L}|} \geq \frac{1}{\frac{2n}{\lambda}+1-\frac{3}{\lambda}} = \frac{\frac{\lambda}{n}}{2+\frac{\lambda - 3}{n}}  \geq \frac{1}{2+\frac{\lambda - 3}{n}} \cdot p^*(G) ,
	\end{align*}
	where the first equality is due to the fact that any leaf $v \in V$ of the spanning tree $T$ of $G$ through which $\mathcal{L}$ was created has $\coverage_{\mathcal{L}}(v) = 1$, and therefore there is such a vertex $v$ in $G$ that is covered by exactly one $\lambda$-subgraph; and the last inequality is due to the fact that $p^*(G) \leq \lambda/n$ for any graph $G$, where $p^*(G)$ is the MaxMin probability of $G$.
	
	The above inequality implies that if the defender chooses the prescribed strategy the minimum defense ratio cannot be too bad. That is because in the worst case for the defender, each and every attacker will choose a vertex $v'$ on which the aforementioned strategy of the defender results to vertex probability $p'$ (so that the attacker is caught with minimum probability). As a result, the defender will have the minimum possible expected payoff which is $p' \cdot k$. Thus, for the constructed defend strategy and any combination of attack strategies, the resulting strategy profile $S$ yields defense ratio:
	\begin{align*}
	\DR{G}{S} \leq \frac{k}{p' \cdot k} \leq \left( 2+\frac{\lambda - 3}{n} \right) \cdot \frac{1}{p^*(G)} = \left( 2+\frac{\lambda - 3}{n} \right) \cdot \DR{G}{S^*},
	\end{align*} 
	where the last equality is due to Lemma \ref{lem:char_DR}.
	
	With respect to the running time, notice that Step~\ref{step:span} of Algorithm~\ref{alg:approximation_of_p*} can be executed in time $O(|V|+|E(G)|) = O(|E(G)|)$. Step~\ref{step:collection} can be executed in time $O(|V|+|E(T)|) = O(|V|)$. Finally, Step \ref{step:prob} can be executed in constant time. Therefore, the total running time of Algorithm \ref{alg:approximation_of_p*} is $O(|E(G)|)$.
\end{proof}

\begin{corollary}
	For any graph $G$ there is a polynomial (in both $n$ and $\lambda$) time approximation algorithm (Algorithm \ref{alg:approximation_of_p*}) with approximation factor $1/\left( 2+\frac{\lambda - 3}{n} \right)$ for the computation of $p^*(G)$. 
\end{corollary}

The merit of finding a probability $p'$ that approximates (from below) $p^*(G)$ for a given graph $G$ through an algorithm such as Algorithm \ref{alg:approximation_of_p*} is in guaranteeing to the defender that, no matter what the attackers play, she always ``catches'' at least a portion $p'$ of them in expectation, where the best portion is $p^*(G)$ in an equilibrium. Algorithm \ref{alg:approximation_of_p*} guarantees that the defender catches at least $1/\left( 2+\frac{\lambda - 3}{n} \right)$ of the attackers in expectation.

\section{Bounds on the Price of Defense}

\begin{theorem}\label{thm:low_b-PoD}
	The \PoD($\lambda$) is lower bounded by $\floor{\frac{2(n-1)}{\lambda}}$ and $\floor{\frac{2(n-1)}{\lambda + 1}}$ for $\lambda$ even and odd respectively, when $\lambda \in \{2, 3, \dots, n-1\}$.
\end{theorem}

\begin{proof}
	We will prove the statement by showing that for any given $n$ and $\lambda \in \{ 2, 3,\dots, n-1 \}$, there exists a graph $G = (V,E)$ on $n$ vertices that requires (at least) some number roughly $b = \floor{\frac{2(n-1)}{\lambda + 1}}$ of $\lambda$-subgraphs to be covered and additionally this graph's structure achieves $p^*(G)$ for the uniform defense strategy, i.e. each $\lambda$-subgraph is assigned equal probability $1/b$.
	
	The graph we construct is the following. First, consider a line graph with $\sigma$ vertices, where $\sigma = \ceil{\frac{\lambda}{2}}$. Keep a {\em central vertex} to use later, and using only $n-1$ vertices, create as many {\em complete lines} with $\sigma$ vertices as possible, i.e. $b = \floor{\frac{n-1}{\sigma}}$. Create another {\em incomplete line} (if needed) with strictly less than $\sigma$ vertices using the remaining ones $n-1 - b \cdot \sigma$. Now draw an edge from the central vertex to a single leaf of each of the constructed lines. For examples of the construction of $G$ in each of the below three cases, see Figures \ref{fig:case1}, \ref{fig:case2a}, and \ref{fig:case2b}.
	
	Consider now a defense strategy $q:=(q_1, q_2, \dots, q_{\theta}) \in \Delta_{\theta}$ and the vertex probabilities $p_1, p_2, \dots, p_n$ it induces on the vertices of $G$. 

	\textbf{Case 1: $\lambda$ is even.} In this case $\sigma = \lambda / 2$ and observe that the diameter of this graph $G$ is equal to $\lambda +1$, therefore no $\lambda$-subgraph that covers a leaf of a complete line can cover a leaf of another complete line. Also, any $\lambda$-subgraph that covers a leaf of a complete line can cover the whole incomplete line. Therefore, this graph can be covered by $b$ $\lambda$-subgraphs but no less. Assume that $q$ covers $G$, i.e. $p_i > 0, \forall i \in [n]$, and let us focus on the set $V_{com}$ of leaves of the complete lines of $G$, where $|V_{com}| = b$ as argued earlier, and denote $V_{com}$ by $[b]$. Consider the vertex probabilities $p_i$, $i \in [b]$, and note that $\sum_{i \in [b]} p_i \leq 1$ where strict inequality holds for the case where there exists some pure strategy $L_j \in supp(q)$ such that $L_j \cap V_{com} = \emptyset$. Then for $p' := \min_{i \in [b]} p_i$ it holds that $p' \leq 1/b$, otherwise $p_{i} > 1/b$, $\forall i \in [b]$ and therefore $\sum_{i \in [b]} p_i > 1$ which is a contradiction. Also, for $p_i = 1/b$, $\forall i \in [b]$, it is $p' = 1/b$, which yields $p^*(G) := \max_{q \in \Delta_{\theta}} p' = 1/b$. 

	\textbf{Case 2: $\lambda$ is odd.} In this case $\sigma = (\lambda + 1) / 2$ and the diameter of $G$ equals $\lambda + 2$, therefore no $\lambda$-subgraph that covers a leaf of a complete line can cover a leaf of another complete line.  
	\begin{itemize}
		\item \textbf{Subcase (a): $\sigma - \left( n-1 - b \cdot \sigma \right) \neq 1$.} Any $\lambda$-subgraph that covers a leaf of a complete line can cover the whole incomplete line. Therefore, this graph can be covered with $b$ $\lambda$-subgraphs but no less. Following the analysis of Case 1, it is $p^*(G) := \max_{q \in \Delta_{\theta}} p' = 1/b$.
	 	\item \textbf{Subcase (b): $\sigma - \left( n-1 - b \cdot \sigma \right) = 1 $.} No $\lambda$-subgraph that covers a leaf of a complete line can cover the leaf of the incomplete line. Therefore, this graph can be covered by $b+1$ $\lambda$-subgraphs but no less. Following similar analysis as that of Case 1, where instead of $V_{com}$ we have $V_{com} \cup \{ v_{inc} \}$ where $v_{inc}$ is the leaf of the incomplete line, and instead of $b$ we have $b+1$, we conclude that $p^*(G) := \max_{q \in \Delta_{\theta}} p' = 1/(b+1)$.
	\end{itemize}
	
	For Case 1, and Case 2(a), since each of the leaves of the $b$ complete lines have vertex probability $1/b$, the defense strategy $q^*$ with probability $q^*_i = 1/b$ assigned to the respective pure strategy $L_i , i \in [b]$ that contains vertex $i \in [b]$, yields $p^*(G)$. For Case 2(b), since each of the leaves of the $b$ complete lines and the leaf $v_{inc}$ of the incomplete line have vertex probability $1/b$, the defense strategy $q^*$ with probability $q^*_i = 1/(b+1)$ assigned to the respective pure strategy $L_i , i \in [b] \cup \{ v_{inc} \}$ that contains vertex $i \in [b] \cup \{ v_{inc} \}$, yields $p^*(G)$.
	
	By the above values of $p^*(G)$ and Lemma \ref{lem:char_DR} the proof of the theorem is complete.
	
	\begin{figure}[h!]
		\begin{center}
			\includegraphics[scale=0.41]{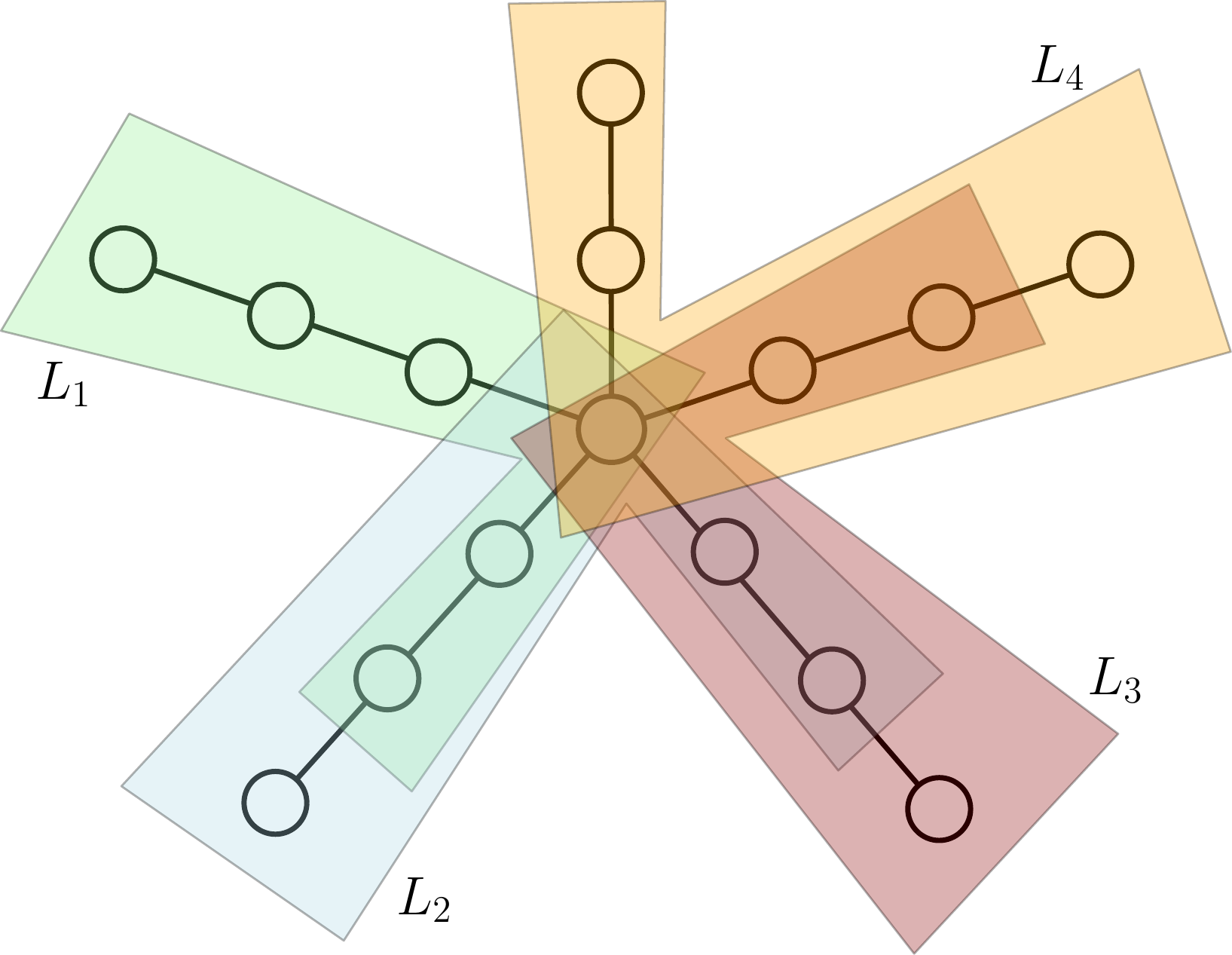}
		\end{center}
		\caption{An example of \textbf{Case 1} of Theorem \ref{thm:low_b-PoD}, where $n=15$ and $\lambda=6$. Here, graph $G$ has $\sigma = 3$ and $b = 4$. The $\lambda$-subgraphs $L_1, L_2, L_3, L_4$ that constitute the support of a best-defense strategy are shown with various colors.}\label{fig:case1}
	\end{figure}
	
	\phantom{ }
	
	\begin{figure}[h!]
		\centering
		\begin{minipage}{0.47\textwidth}
			\centering
			\includegraphics[width=0.9\linewidth]{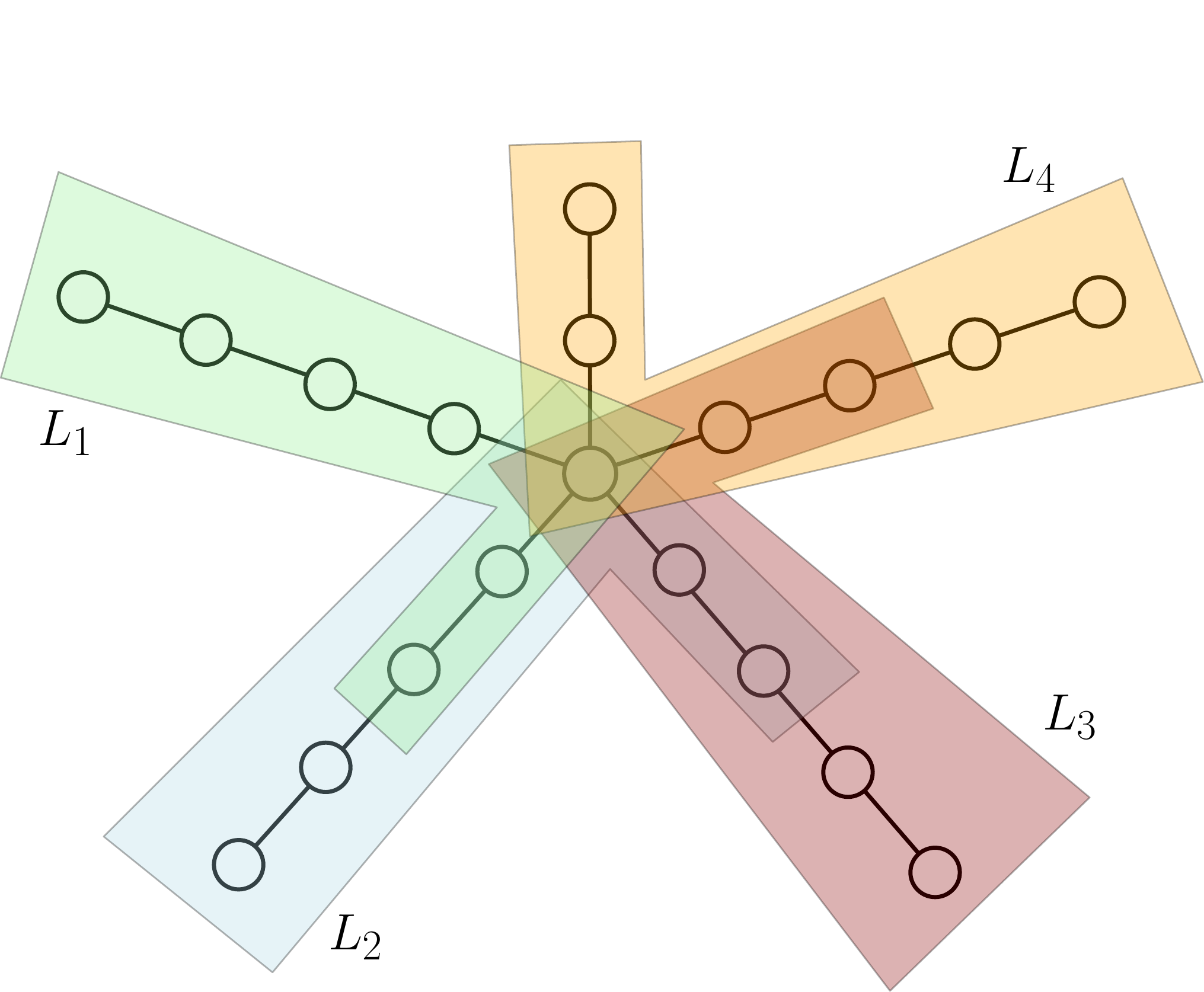}
			\caption{An example of \textbf{Case 2(a)} of Theorem \ref{thm:low_b-PoD}, where $n=19$ and $\lambda=7$. Here, graph $G$ has $\sigma = 4$ and $b = 4$. The $\lambda$-subgraphs $L_1, L_2, L_3, L_4$ that constitute the support of a best-defense strategy are shown with various colors.}\label{fig:case2a}
		\end{minipage}%
		\hfill
		\begin{minipage}{0.47\textwidth}
			\centering
			\includegraphics[width=0.9\linewidth]{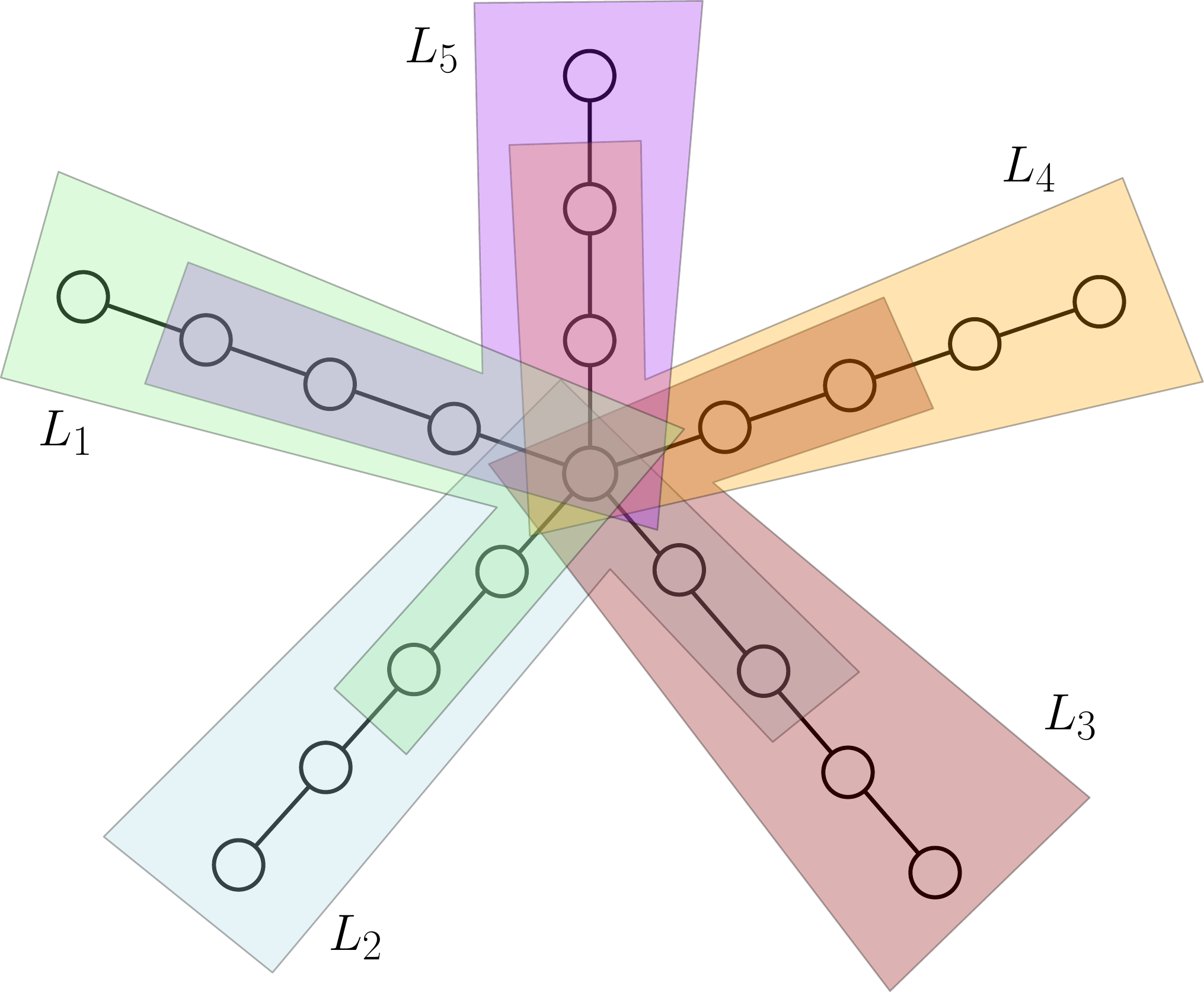}
			\caption{An example of \textbf{Case 2(b)} of Theorem \ref{thm:low_b-PoD}, where $n=20$ and $\lambda=7$. Here, graph $G$ has $\sigma = 4$ and $b = 4$. The $\lambda$-subgraphs $L_1, L_2, L_3, L_4, L_5$ that constitute the support of a best-defense strategy are shown with various colors.}\label{fig:case2b}
		\end{minipage}
	\end{figure}

	\phantom{ }
\end{proof}

\begin{corollary}
	For any given $n$ and $2 \leq \lambda \leq n-1 $, it holds that $\floor{\frac{2(n-1)}{\lambda + 1}} \leq $\PoD$(\lambda) \leq  \frac{2(n-1) + \lambda - 1}{\lambda}$. Furthermore, for the trivial cases $\lambda \in \{ 1, n \}$ it is $\PoD(1)=n$ and $\PoD(n)=1$.
\end{corollary}

\begin{proof}
	For the lower bound for $2 \leq \lambda \leq n-1$, Theorem \ref{thm:low_b-PoD} shows that for given $n$ and $\lambda$ there exists a graph $G$ with particular (very small) $p^*(G)$, and according to Lemma \ref{lem:char_DR} this yields the corresponding (great) best defense ratio. The upper bound is due to Theorem \ref{thm:appr_def_rat}. For the cases $\lambda=1$ and $\lambda = n$, observe that the defender's action set is $D = \{ \{ i \} | i \in V \}$ and $D = V$ respectively, therefore $p^*(G) = 1/n$ and $p^*(G) = 1$ respectively, and again from Lemma \ref{lem:char_DR} we get the values in the statement of the corollary.
\end{proof}



\bibliographystyle{splncs04}
\bibliography{references}

\end{document}